\documentclass[10pt]{extarticle}

\usepackage[english]{babel}
\usepackage[ansinew]{inputenc}
\usepackage{amsmath}
\usepackage{amsfonts}
\usepackage{amssymb}
\usepackage{amsthm}
\usepackage{graphicx}
\usepackage{fancybox}
\usepackage{calc}
\usepackage{tabularx}
\usepackage{tikz}
\usepackage{array}
\usepackage{lmodern}
\usepackage{mathabx}
\usepackage{fancybox}
\usepackage{caption}
\usepackage{fullpage}

\usepackage[margin=1.6in]{geometry}

\newcommand\numberthis{\addtocounter{equation}{1}\tag{\theequation}}
\newcommand{\E}{\mathbb{E}}

\newtheorem{Thm}{Theorem}
\newtheorem{Lem}[Thm]{Lemma}

\newtheorem{Def}{Definition}

\newtheorem{Exm}{Example}

\newtheorem{Rem}{Remark}

\newtheorem{Ntn}{Notation}

\title{Private Stream Aggregation Revisited
\thanks{The research was supported by the DFG Research Training Group GRK $1817/1$}
}
\author{
Filipp Valovich, Francesco Ald\`{a}\\ Horst G\"{o}rtz Institute for IT Security\\ Faculty of Mathematics\\ Ruhr-Universit\"{a}t Bochum, Universit\"{a}tsstrasse 150, 44801 Bochum, Germany\\ Email: \{filipp.valovich, francesco.alda\}@rub.de
}
\date{}

\begin{document}

\maketitle

\pagenumbering{gobble}
\pagestyle{empty}


\begin{abstract} \noindent In this work, we investigate the problem of private statistical analysis in the distributed and semi-honest setting. In particular, we study properties of Private Stream Aggregation schemes, first introduced by Shi et al. \cite{2}. These are computationally secure protocols for the aggregation of data in a network and have a very small communication cost. 
We show that such schemes can be built upon any key-homomorphic \textit{weak} pseudo-random function. 
Thus, in contrast to the aforementioned work, our security definition can be achieved in the \textit{standard model}. In addition, we give a computationally efficient instantiation of this protocol based on the Decisional Diffie-Hellman problem. 
Moreover, we show that every mechanism which preserves $(\epsilon,\delta)$-differential privacy provides \textit{computational} $(\epsilon,\delta)$-differential privacy when it is executed through a Private Stream Aggregation scheme. Finally, we introduce a novel perturbation mechanism based on the \textit{Skellam distribution} that is suited for the distributed setting, and compare its performances with those of previous solutions.
\end{abstract}

\section{Introduction}

The framework of statistical disclosure control aims at providing strong privacy guarantees for the records stored in a database while enabling accurate statistical analyses to be performed. In recent years, \textit{differential privacy} has become one of the most important paradigms for privacy-preserving statistical analyses. Generally, the notion of differential privacy is considered in the centralised setting where we assume the existence of a \textit{trusted curator} \cite{16, 24, 8, 9} who collects data in the clear, perturbs it properly (e.g. by adding Laplace noise) and publishes it. In this way, the output statistics are not significantly influenced by the presence (resp. absence) of a particular record in the database.\\
In this work, we study how to preserve differential privacy when we cannot rely on a trusted curator. In this so-called \textit{distributed setting}, the users have to send their data to an untrusted aggregator. Preserving differential privacy and achieving high accuracy in the distributed setting is of course harder than in the centralised setting, since the users have to execute a perturbation mechanism on their own. In order to achieve the same accuracy as provided by well-known techniques in the centralised setting, the work by Shi et al. \cite{2} introduces the \textit{Private Stream Aggregation} (PSA) scheme, a cryptographic protocol which enables each user to securely send encrypted time-series data to an aggregator. The aggregator is then able to decrypt the aggregate of all data in each time step, but cannot retrieve any further information about the individual data. Using such a protocol, the task of perturbation can be split among the users, such that differential privacy is preserved \textit{and} high accuracy is guaranteed. For a survey of applications of this protocol, we refer to \cite{2}.\\
In \cite{2}, a PSA scheme for sum queries is provided and some strong security guarantees under the Decisional Diffie-Hell\-man assumption are shown. However, this instantiation has some limitations. First, the security only holds in the random oracle model; second, its decryption algorithm requires the solution of the discrete logarithm in a given range, which can be very time-consuming if the number of users and the plaintext space are large. 
Moreover, since a PSA scheme provides \textit{computational} security, the perturbation mechanism in use can only provide a computational version of differential privacy, a notion first introduced by Mironov et al. \cite{15}. In \cite{2}, however, a connection between the security of a PSA scheme and differential privacy is not explicitly shown. In a subsequent work by Chan et al. \cite{3}, this connection is still not completely established, since the polynomial-time reduction between an attacker against a PSA scheme and a database distinguisher is missing.\\

\noindent\textbf{Objectives.} In order to overcome the limitations of the construction in \cite{2}, in this work we address the following problems.
\begin{itemize}
 \item We want to give \textit{sufficient conditions} for a PSA scheme which has a security guarantee in the \textit{standard model}.
 \item According to these conditions, we want to construct a concrete instantiation consisting of efficient algorithms, even when the number of users and the plaintext space become large.
 \item We aim at showing that an information-theoretical differentially private mechanism preserves \textit{computational differential privacy} when it is executed through a (computationally) secure PSA scheme. 
 \item We want to investigate differentially private mechanisms suitable for an execution through a PSA scheme.
\end{itemize}

\noindent\textbf{Contributions.} We achieve the aforementioned goals in the following manner. In order to derive sufficient conditions for PSA schemes with a certain security guarantee, we lower the requirements of Aggregator Obliviousness from \cite{2} by abrogating the attacker's possibility to \textit{adaptively compromise} users during the execution of a PSA scheme with time-series data. We show that a PSA scheme for achieving this lower security level can be built upon any \textit{key-homomorphic weak pseudo-random function}. Since weak pseudo-randomness can be achieved in the standard model, this condition also enables secure schemes in the standard model. In particular, we can build a key-homomorphic weak pseudo-random function based on the Decisional Diffie-Hell\-man assumption in the group of quadratic residues modulo a \textit{squared safe prime}. This function is used for the construction of a PSA scheme for sum queries. By comparing the running times and practical performances of our PSA scheme and the one given by Shi et al. \cite{2} at the same security level, we find that our solution provides a significant speed-up for decryption when the plaintext space is large while decelerating the encryption only by a constant factor. 

Reduction-based security proofs for cryptographic schemes usually require an attacker in the corresponding security game to send two different plaintexts (or plaintext collections) to a challenger. The adversary receives then back a ciphertext which is the encryption of one of these collections and has to guess which one it is. In any security definition for a PSA scheme, these collections must satisfy a particular requirement (i.e. they must lead to the same aggregate), since the attacker has the capability to decrypt them (different aggregates would make the adversary's task trivial). In general, however, this requirement cannot be satisfied in the context of differential privacy. Introducing a novel kind of security reduction which deploys a \textit{biased coin} flip, we can show that, whenever a randomised perturbation procedure is involved in a PSA scheme, the requirement of having collections with equal aggregate can be abolished. This result can be generalised to any cryptographic scheme with such a requirement. Using this property, we are able to show that if a mechanism preserves differential privacy, then it preserves computational differential privacy when it is used as a randomised perturbation procedure in a PSA scheme.


Finally, we compare three mechanisms: the Geometric mechanism from \cite{2}, the Binomial mechanism from \cite{14} and the \textit{Skellam mechanism} introduced in this work. All three mechanisms preserve differential privacy and make use of discrete probability distributions. Therefore, they are well-suited for an execution through a PSA scheme. For generating the right amount of noise among all users, these mechanisms apply two different approaches. While in the Geometric mechanism, with high probability, only one user generates the noise necessary for differential privacy, the Binomial and Skellam mechanisms allow all users to generate noise of small variance, that sums up to the required value for privacy. We show that for high privacy levels, the theoretical error bound of the Skellam mechanism is slightly better than that of the other two. At the same time, we provide experimental results showing that the Geometric and Skellam mechanisms have a comparable accuracy in practice, while beating the one of the Binomial mechanism.\\

\noindent\textbf{Related Work.} 
As pointed out above, our contributions are mostly related to the work of Shi et al. \cite{2} and Chan et al. \cite{3}.
Privacy-preserving aggregation of time-series data in the presence of an untrusted aggregator has also been studied in various other works, e.g. \cite{4, 37, 33, 34, 38}. Beimel et al. \cite{4} and Eigner et al. \cite{37} show that secure multi-party computation techniques can be used for data aggregation under differential privacy. These techniques usually have a high communication cost, whereas PSA requires each user to send exactly one message per time-step. 
The protocol given by Rastogi et al. \cite{19} is based on the threshold Paillier cryptosystem. It requires an extra round of interaction between the users and the aggregator in every time-step in order to decrypt the sum queries. In contrast, PSA requires the users to interact with the aggregator only for sending the ciphertexts. \'{A}cs et al. \cite{36} use an additive homomorphic encryption scheme for sending time-series data, but it requires the generation of a pair of encryption/decryption keys for each pair of users. Moreover, reuse of key pairs for different time-steps potentially leads to security breaches. In a PSA scheme each user gets only one encryption key, which can be securely used for \textit{every} time-step. Using additive homomorphic encryption, Rieffel et al. \cite{35} construct a scheme which does not require extra rounds of interaction, but is not fully resistant against collusions and the cost of computation and storage is roughly equal to the number of compromised users that is tolerable by the system. 
Li et al. \cite{33,34} use the homomorphic encryption scheme given by Castelluccia et al. \cite{31,32} in order to construct an efficient protocol for sending data in mobile sensing applications. 
This scheme is resistant against collusions, but each user has to store multiple keys, depending on the number of compromised users in the network. Moreover, for encryption and decryption the scheme requires the computation of as many pseudo-random values as the number of keys in the network, making the computational effort for the analyst rather high. Thus, the costs of this scheme depend on the number of compromised users. 
A PSA scheme is fully resistant against \textit{any} number of collusions and furthermore, we provide a solution, where the computation and storage costs are independent of the number of users. Joye et al. \cite{38} provide a protocol with the same security guarantees as in \cite{2} in the random oracle model. The security of their scheme relies on the DCR assumption (rather than DDH as in \cite{2}) and as a result, in the security reduction they can remove a factor which is cubic in the number of users. However, their scheme involves a trusted party for setting some public parameters. In this work we provide an instantiation of our \textit{generic} PSA construction, which is similar to the one in \cite{38} but relies on the DDH assumption. While in our generic security reduction we cannot avoid the cubic factor in the number of users, our construction \textit{does not} involve any trusted party and has security guarantees in the standard model.

Another series of works deals with a distributed generation of noise for preserving differential privacy. Dwork et al. \cite{14} consider the Gaussian distribution for splitting the task of noise generation among all users. Their proposed scheme requires more interactions between the users than our solution. In \cite{36}, privacy-preserving data aggregation is applied to smart metering and the generation of Laplace noise is performed in a distributed manner, since each meter simply generates the difference of two Gamma distributed random variables as a share of a Laplace distributed random variable. In \cite{19} each user generates a share of Laplace noise by generating a vector of four Gaussian random variables. For a survey of the mechanisms given in \cite{19} and \cite{36}, we refer to \cite{21}. However, the aforementioned mechanisms generate noise drawn according to continuous distributions, but for the use in a PSA scheme, discrete noise is required. Therefore, we consider proper discrete distributions and compare their performances for private statistical analyses.

\section{Preliminaries}

\subsection{Problem statement}

In this work, we consider a distributed and semi-honest setting where $n$ users are asked to participate in some statistical analyses but do not trust the data analyst (or aggregator), who is assumed to be honest but curious. Therefore, the users cannot provide their own data in the clear. Moreover, they communicate solely and independently with the untrusted aggregator, who wants to analyse the users data by means of queries in time-series and aims at obtaining answers as accurate as possible. More specifically, assume that the data items belong to a data universe $\mathcal{D}$. For a sequence of time-steps $t\in T$, where $T$ is a discrete time period, the analyst sends queries which are answered by the users in a distributed manner. Each query is modeled as a function $f:\mathcal{D}^n\to \mathcal{O}$ for a finite or countably infinite set of possible outputs (i.e. answers to the query) $\mathcal{O}$.\\ 
We also assume that some users may act in order to compromise the privacy of the other participants. More precisely, we assume the existence of a publicly known constant $\gamma\in(0,1]$ which is the a priori estimate of the lower bound on the fraction 
of non-compromised users who honestly follow the protocol and want to release useful information about their data (with respect to a particular query $f$), while preserving $(\epsilon,\delta)$-differential privacy. 
The remaining $(1-\gamma)$-fraction of users is assumed to be compromised. Compromised users honestly follow the protocol but are aimed at violating the privacy of non-compromised users.
For that purpose, these users form a coalition with the analyst and send her auxiliary information, e.g. their own data in the clear.\\ 
For computing the answers to the aggregator's queries, a special cryptographic protocol, called Private Stream Aggregation (first introduced in \cite{2}), is used by \textit{all} users. In connection with a perturbation mechanism, this scheme assures that the analyst is only able to learn a noisy aggregate of the users' data (as close as possible to the real answer $f(D)$) and nothing else. In contrast to common secure multi-party techniques \cite{25, 6, 20}, this protocol requires each user to send to the analyst only one message per query. 

\subsection{Definitions}

We consider a database as an element $D\in\mathcal{D}^n$, where $\mathcal{D}$ is the data universe and $n$ is the number of users. Since $D$ may contain sensitive information, the users want to protect their privacy. Therefore, a privacy-preserving mechanism must be applied. Unless stated differently, we always assume that a mechanism is applied in the distributed setting. Differential privacy \cite{8} is a well-established notion for privacy-preserving statistical analyses. We recall that a randomised mechanism preserves differential privacy if its application on two adjacent databases, i.e. databases which differ in one entry only, leads to close distributions of the output.

\begin{Def}[Differential Privacy~\cite{8}]
Let $\mathcal{R}$ be a (possibly infinite) set and let $n\in\mathbb{N}$. A randomised mechanism $\mathcal{A}:\mathcal{D}^n\to\mathcal{R}$ preserves $(\epsilon,\delta)$-differential privacy (short: \mbox{\upshape\sffamily DP}), if for all adjacent databases $D_0, D_1\in\mathcal{D}^n$ and all $R\subseteq\mathcal{R}$:
\[\Pr[\mathcal{A}(D_0)\in R]\leq e^\epsilon\cdot \Pr[\mathcal{A}(D_1)\in R]+\delta.\]
The probability space is defined over the randomness of $\mathcal{A}$.
\end{Def}

The additional parameter $\delta$ is necessary for mechanisms which cannot preserve $\epsilon$-\mbox{\upshape\sffamily DP} (i.e. $(\epsilon,0)$-\mbox{\upshape\sffamily DP}) for certain cases. However, if the probability that these cases occur is bounded by $\delta$, then the mechanism preserves $(\epsilon,\delta)$-\mbox{\upshape\sffamily DP}.

In the literature, there are well-established mechanisms for preserving differential privacy, e.g. the \textit{Laplace mechanism} \cite{8} and the \textit{Exponential mechanism} \cite{9}. In order to privately evaluate a query, these mechanisms draw noisy values according to some distribution depending on the query's global sensitivity.

\begin{Def}[Global Sensitivity~\cite{8}]
The global sensitivity $S(f)$ of a query\linebreak $f:\mathcal{D}^n\to\mathbb{R}^k$ is defined as
\[S(f)=\max_{D_0,D_1\mbox{\scriptsize \,\, adjacent}}||f(D_0)-f(D_1)||_1.\]
\end{Def}

In particular, we will consider sum queries $f_{\mathcal{D}}:\mathcal{D}^n\to\mathbb{Z}$ defined as $f_{\mathcal{D}}(D):=\sum_{i=1}^n d_i$,\linebreak for $D=(d_1,\ldots,d_n)\in\mathcal{D}^n$ and $\mathcal{D}\subseteq\mathbb{Z}$.\\
For measuring how well the output of a mechanism estimates the real data with respect to a particular query, we use the notion of $(\alpha,\beta)$-accuracy.

\begin{Def}[Accuracy~\cite{17}]
The output of a mechanism $\mathcal{A}$ achieves $(\alpha,\beta)$-accuracy for a query $f:\mathcal{D}^n\to\mathbb{R}$ if for all $D\in\mathcal{D}^n$:
\[\Pr[|\mathcal{A}(D)-f(D)|\leq\alpha]\geq 1-\beta.\]
The probability space is defined over the randomness of $\mathcal{A}$.
\end{Def}

The use of a cryptographic protocol for transferring data 
provides a computational security level. If such a protocol is applied for preserving differential privacy, this implies that only a computational level of differential privacy can be provided. Our definition of computational differential privacy follows the notion of Chan et al. \cite{3}.

\begin{Def}[Computational Differential Privacy~\cite{3}] 
\mbox{\,\,} Let $\kappa$ be a security parameter and $n\in\mathbb{N}$ with $n=\text{poly}(\kappa)$. A randomised mechanism $\mathcal{A}:\mathcal{D}^n\to\mathcal{R}$ preserves computational $(\epsilon,\delta)$-differential privacy (short: \mbox{\upshape\sffamily CDP}), if for all adjacent databases $D_0, D_1\in\mathcal{D}^n$ and all probabilistic polynomial-time distinguishers $\mathcal{D}_{\mbox{\scriptsize\upshape\sffamily CDP}}$:
\[\Pr[\mathcal{D}_{\mbox{\scriptsize\upshape\sffamily CDP}}(1^\kappa,\mathcal{A}(D_0))=1]\leq e^\epsilon\cdot\Pr[\mathcal{D}_{\mbox{\scriptsize\upshape\sffamily CDP}}(1^\kappa,\mathcal{A}(D_1))=1]+\delta+\text{\upshape\sffamily neg}(\kappa),\]
where $\text{\upshape\sffamily neg}(\kappa)$ is a negligible function in $\kappa$. The probability space is defined over the randomness of $\mathcal{A}$ and $\mathcal{D}_{\mbox{\scriptsize\upshape\sffamily CDP}}$.
\end{Def}

The notion of computational differential privacy is a natural computational-\linebreak indistinguishability-extension of the infor\-mation-theoretical definition. The advantage is that preserving differential privacy only against bounded attackers helps to substantially reduce the error of the answer provided by the mechanism. In Section \ref{psa}, we investigate how to obtain a computationally secure protocol which allows the analyst to compute only the aggregate of all users' data and no further information. The scheme for sum queries we are going to construct uses a special mapping into a group, which we define formally.

\begin{Def}[$v$-isomorphic embedding] An \,\,\,injec\-tive mapping $\varphi:\{-v,\ldots,v\}\to V$, where $(V,\circ)$ is a group, is a $v$-isomorphic embedding if for all $n\in\mathbb{N}$ and all finite sequences $(a_i)_{i=1,\ldots,n}$ of values in $\{-v,\ldots,v\}$ with $\left|\sum_i a_i\right|\leq v$:
\[\varphi\left(\sum_{i=1}^n a_i\right)=\varphi(a_1)\circ\ldots\circ\varphi(a_n).\]
\end{Def}

From this definition it is clear that a $v$-isomorphic embedding is also $v^\prime$-isomorphic for every integer $0<v^\prime\leq v$. In the analysis of the secure protocol, we furthermore make use of the following definition.

\begin{Def}[Weak PRF~\cite{26}] Let $\kappa$ be a security parameter. Let $A,B,C$ be sets. A family of functions 
\[\mathcal{F}=\{\text{\upshape\sffamily F}_a\,|\,\text{\upshape\sffamily F}_a: B\to C\}_{a\in A}\] 
is called a weak pseudo-random function (PRF) family if for all probabilistic polynomial-time algorithms $\mathcal{D}_{\mbox{\scriptsize PRF}}^{\mathcal{O}(\cdot)}$ with oracle access to $\mathcal{O}(\cdot)$ (where $\mathcal{O}(\cdot)\in\{\text{\upshape\sffamily F}_a(\cdot),\text{\upshape\sffamily rand}(\cdot)\}$) on any polynomial number of uniformly chosen inputs, we have:
\[|\Pr[\mathcal{D}_{\mbox{\scriptsize PRF}}^{\text{\upshape\sffamily F}_a(\cdot)}(\kappa)=1]-\Pr[\mathcal{D}_{\mbox{\scriptsize PRF}}^{\text{\upshape\sffamily rand}(\cdot)}(\kappa)=1]|\leq\text{\upshape\sffamily neg}(\kappa),\]
where $a\in_R A$ and $\text{\upshape\sffamily rand}\in_R\{f\,|\,f:B\to C\}$ is a random mapping from $B$ to $C$.
\end{Def}

\subsection{Mechanism overview}\label{mechov}

In this work we prove the following result by showing the connection between a key-homomorphic weak pseudo-random function and a differentially private mechanism for sum queries.


\begin{Thm}\label{mainthm} Let $\epsilon>0$, $w<w^\prime\in\mathbb{Z}$, $m,n\in\mathbb{N}$ with $\max\{|w|,|w^\prime|\}<m$. Let\linebreak $\mathcal{D}=\{w,\ldots,w^\prime\}$ and $f_{\mathcal{D}}$ be a sum query. If there exist groups $G^\prime\subseteq G$, a key-homomorphic weak pseudo-random function family mapping into $G^\prime$ and an efficiently computable and efficiently invertible $mn$- isomorphic embedding \[\varphi:\{-mn,\ldots,mn\}\to G,\] then there exists an efficient mechanism for $f_{\mathcal{D}}$ that preserves $(\epsilon,\delta)$-\mbox{\upshape\sffamily CDP} for any constant $0<\delta<1$ with an error bound of $O(S(f_{\mathcal{D}})/\epsilon)$ and requires each user to send exactly one message.
\end{Thm}

As already described, we want the untrusted analyst to be able to learn some aggregated statistics $f_{\mathcal{D}}(D)$ but no additional information about each user's data. Assume that we can design a cryptographic protocol that achieves the aforementioned goal. If we furthermore aim at preserving $(\epsilon,\delta)$-\mbox{\upshape\sffamily DP}, it would be sufficient to add a single copy of (properly distributed) noise $Y$ to the value $f_{\mathcal{D}}(D)$. Since we cannot add such noise once the aggregate has been computed, the users have to generate and add noise to their original data in such a way that the sum of the perturbations has the same distribution as $Y$. For this purpose, we see two different approaches. In the first approach, there is a small probability (depending on the fixed parameter $\gamma$) for each user to add noise sufficient to preserve the privacy of the entire statistics. This probability is calibrated in such a way only one of the $n$ users is expected to add noise at all. Shi et al. \cite{2} investigate this method using the Geometric mechanism. In the second approach, each user generates noise of small variance (again depending on $\gamma$), such that the sum of all noisy terms has enough variance for preserving differential privacy. For this aim, we need discrete probability distributions which are closed under convolution and are known to provide differential privacy. The Binomial mechanism \cite{14} and the Skellam mechanism introduced in this work serve these purposes. In both approaches, the error which is introduced is reasonably small and similar theoretical bounds can be provided. For details, see Section \ref{dpmech}.

For a particular time-step, let the users' values be of the form $x_i=d_i+r_i$, $i=1,\ldots,n$, where $d_i\in\mathcal{D}$ is the original data of the user $i$ and $r_i$ is her noisy value. In the privacy analysis, it is reasonable to assume that $r_i=0$ for the $(1-\gamma)\cdot n$ compromised users, since this can only increase their chances to infer some information about the non-compromised users. In order to send the values to the data analyst, the users perform a PSA scheme. First, each user encrypts her own time-series data and sends the ciphertexts to the data analyst. After a distributed key exchange, the evaluation of a single query (i.e. a query analysed in one time-step) requires each user to send exactly one message. The data analyst appropriately aggregates the ciphertexts of all users for a particular time-step and then decrypts the sum of the users' values $\sum_{i=1}^n x_i$. From the ciphertexts, the data analyst is not able to leach any additional information about the values of the users, except for the auxiliary information obtained from the compromised users. In this way, there is no privacy-breach if only one user adds the entirely needed noise (first approach) or if the non-compromised users generate noise of low variance (second approach), since the single values are encrypted and the analyst cannot learn anything about them, except for their aggregate. Due to the use of a cryptographic protocol, the plaintexts have to be discrete. This is the reason why we use discrete distributions for generating the noisy values $r_i$.\\
The perturbation of data potentially yields larger values $x_i$ due to the (possibly) infinite domain of the underlying probability distribution. Depending on the variance, we therefore need to choose a sufficiently large interval $\widehat{\mathcal{D}}=\{-m,\ldots,m\}$ as plaintext space, where $m>\max\{|w|,|w^\prime|\}$ such that $|x_i|\leq m$ for all $i=1,\ldots,n$ with high probability. In the following, we always assume that $\mathcal{D}$ is a subinterval of $\widehat{\mathcal{D}}$.\\
Since the protocol used for the data transmission is computationally secure, the entire mechanism preserves $(\epsilon,\delta)$-\mbox{\upshape\sffamily CDP}.


\section{Private Stream Aggregation}\label{psa}

In this section, we define the Private Stream Aggregation scheme and give a security definition for it. Thereby, we mostly follow the concepts introduced by Shi et al. \cite{2}, though we deviate in a few points. Afterwards, we give a condition for the existence of secure PSA schemes. Moreover, we give a concrete and efficient instantiation of a secure PSA scheme in the standard model.

\subsection{The definition of Private Stream Aggregation and its security}

\noindent\textbf{Private Stream Aggregation.} A PSA scheme is a protocol for safe distributed time-series data transfer which enables the receiver to learn only the aggregate $f(\widehat{D})$ of a query $f:\widehat{\mathcal{D}}^n\to\mathcal{O}$ over some distributed (and possibly perturbed) database $\widehat{D}\in\widehat{\mathcal{D}}^n$. 
Such a scheme needs a key exchange protocol for all $n$ users together with the analyst as a precomputation, and requires each user to send exactly one message per query. For the definition of PSA, we follow \cite{2}.

\begin{Def}[Private Stream Aggregation]
Let $\kappa$ be a security parameter and $n\in\mathbb{N}$ with $n=\text{poly}(\kappa)$. A Private Stream Aggregation scheme $\Sigma=(\mbox{\upshape\sffamily Setup}, \mbox{\upshape\sffamily PSAEnc}, \mbox{\upshape\sffamily PSADec})$ is defined by three probabilistic polynomial-time Algorithms:
\begin{description}
\item \textbf{\mbox{\upshape \sffamily Setup}}: $(\mbox{\upshape\sffamily pp},T,s,s_1,\ldots,s_n)\leftarrow \mbox{\upshape\sffamily Setup}(1^\kappa)$, where $\mbox{\upshape\sffamily pp}$ are public parameters of the system, $T$ is a set of time-steps and $s,s_1,\ldots,s_n$ are private keys.
\item \textbf{\mbox{\upshape \sffamily PSAEnc}}: For time-step $t\in T$ and all $i=1,\ldots,n$: 
\[c_{i,t}\leftarrow \mbox{\upshape\sffamily PSAEnc}_{s_i}(t,x_i)\mbox{ for a data value } x_i\in\widehat{\mathcal{D}}.\]
\item \textbf{\mbox{\upshape \sffamily PSADec}}: For time-step $t\in T$, ciphertexts $c_{1,t},\ldots,c_{n,t}\in\mathcal{C}$, where $\mathcal{C}$ is the range of $\mbox{\upshape \sffamily PSAEnc}$, and a query $f:\widehat{\mathcal{D}}^n\to\mathcal{O}$ compute
\[f(x'_1,\ldots,x'_n)=\mbox{\upshape\sffamily PSADec}_{s}(t,c_{1,t},\ldots,c_{n,t}).\]
For all $t\in T$ and $x_1,\ldots,x_n\in\widehat{\mathcal{D}}$:
\begin{align*}
&\mbox{\upshape\sffamily PSADec}_{s}(t, \mbox{\upshape\sffamily PSAEnc}_{s_1}(t,x_1),\ldots,\mbox{\upshape\sffamily PSAEnc}_{s_n}(t,x_n))\\
=& f(x_1,\ldots,x_n).
\end{align*}
\end{description}
\end{Def}

The Setup-phase has to be carried out just once and for all, and can be performed with a secure multi-party protocol among all users and the analyst. In all other phases, no communication between the users is needed.\\
The system parameters $\mbox{\upshape\sffamily pp}$ are public and constant for all time-steps with the implicit understanding that they are used in $\Sigma$. Every user encrypts her data $x_i$ with her own secret key $s_i$ and sends the ciphertext to the analyst. If the analyst receives the ciphertexts of \textit{all} users for a time-step $t$, it can compute the aggregate, i.e. the evaluation of the query $f$, of the users' data with the decryption key~$s$.\\

\noindent\textbf{Security.} Since our model allows the analyst to compromise users, the aggregator can obtain auxiliary information about the data of the compromised users or their secret keys. Even then a secure PSA scheme should release no more information than the aggregate of the non-compromised users' data in a single time-step.\\
Informally, a PSA scheme $\Sigma$ is secure if every probabilistic polynomial-time algorithm, with knowledge of the analyst's and compromised users' keys and with adaptive encryption queries, has only negligible advantage in distinguishing between the encryptions of two databases $\widehat{D}_0, \widehat{D}_1$ of its choice, where $f(\widehat{D}_0)=f(\widehat{D}_1)$. We can assume that an adversary knows the secret keys of the entire compromised coalition. 
If the protocol is secure against such an attacker, then it is also secure against an attacker without the knowledge of every key from the coalition. Thus, in our security definition we consider the most powerful adversary. In what follows, let $f_{|_X}:\widehat{\mathcal{D}}^{|X|}\to\mathcal{O}$ denote the evaluation of a query $f:\widehat{\mathcal{D}}^n\to\mathcal{O}$ with input $\widehat{D}\in\widehat{\mathcal{D}}^n$ over a subset $X\subseteq[n]$ of users. The security definition is similar to the one in \cite{2}.

\begin{Def}[Security of PSA]\label{securitygame}
Let $\mathcal{T}$ be a probabilistic polynomial-time adversary for a PSA scheme $\Sigma=(\mbox{\upshape\sffamily Setup}, \mbox{\upshape\sffamily PSAEnc}, \mbox{\upshape\sffamily PSADec})$ and let $f:\widehat{\mathcal{D}}^n\to\mathcal{O}$ be a statistical query over the set $\widehat{\mathcal{D}}$. Let $T$ be the set of time-steps for possible data analyses. We define the following security game between a challenger and the adversary $\mathcal{T}$.
\begin{description}
 \item\textbf{Setup.} The challenger runs the \text{\upshape\sffamily Setup} algorithm on input security parameter $\kappa$ and returns public parameters $\mbox{\upshape\sffamily pp}$, time-steps $T$ with $|T|=\text{poly}(\kappa)$ and secret keys $s,s_1,\ldots,s_n$. It sends $\kappa,\mbox{\upshape\sffamily pp}, T, s$ to $\mathcal{T}$.
\item\textbf{Queries.} The challenger flips a random bit $b\in_R\{0,1\}$. $\mathcal{T}$ chooses $U\subseteq[n]$ and sends it to the challenger which returns $(s_i)_{i\in[n]\setminus U}$.
 $\mathcal{T}$ is allowed to query $(i,t,x_i)$ with $i\in U, t\in T, x_i\in\widehat{\mathcal{D}}$ and the challenger returns \[\mbox{\upshape\sffamily PSAEnc}_{s_i}(t,x_i).\]
\item\textbf{Challenge.} $\mathcal{T}$ chooses $t^*\in T$ such that no encryption query at $t^*$ was made. (If there is no such $t^*$ then the challenger simply aborts.) $\mathcal{T}$ queries two different tuples $(x_i^{[0]})_{i\in U},(x_i^{[1]})_{i\in U}$ with
\[f_{|_U}((x_i^{[0]})_{i\in U})=f_{|_U}((x_i^{[1]})_{i\in U}).\] 
For all $i\in U$ the challenger returns 
\[c_{i,t^*}\leftarrow\mbox{\upshape\sffamily PSAEnc}_{s_i}(t^*,x_i^{[b]}).\]
\item\textbf{Queries.} $\mathcal{T}$ is allowed to make the same type of queries as before with the restriction that no encryption query at $t^*$ can be made.
\item\textbf{Guess.} $\mathcal{T}$ outputs a guess about $b$.
\end{description}
The adversary wins the game if it correctly guesses $b$. A PSA scheme is secure if no probabilistic polynomial-time adversary $\mathcal{T}$ has more than negligible advantage (with respect to the parameter $\kappa$) in winning the above game.
\end{Def}

Encryption queries are made only for $i\in U$, since knowing the secret key for all $i\in[n]\setminus U$ the adversary can encrypt a value autonomously. If encryption queries in time-step $t^*$ were allowed, then no deterministic scheme would be secure. The adversary $\mathcal{T}$ can determine the original data of all $i\in[n]\setminus U$ for every time-step, since it knows $(s_i)_{i\in[n]\setminus U}$. Then $\mathcal{T}$ can compute the aggregate of the non-compromised users' data. For example, when $f=f_{\widehat{\mathcal{D}}}$ is a sum query we have 
$f_{|_U}(\widehat{D})=\mbox{\upshape\sffamily PSADec}_s(t,c_{1,t},\ldots,c_{n,t})-f_{|_{[n]\setminus U}}(\widehat{D})$,
where $\widehat{D}=(x_1,\ldots,x_n)$ and $c_{i,t}$ is the encryption of $x_i$ for all $i\in[n]$. On the other hand, if there is a user's ciphertext which $\mathcal{T}$ does not receive, then it cannot obtain the aggregate for the correspondent time-step.\\ 
The security definition indicates that $\mathcal{T}$ cannot distinguish between the encryptions of two different data collections $(x_i^{[0]})_{i\in U},(x_i^{[1]})_{i\in U}$ with the same aggregate at time-step $t^*$. For proving that a secure PSA scheme can be used for computing differentially private statistics with small error, we have to slightly modify the security game such that an adversary may choose adjacent (and non-perturbed) databases, as it is required in the definition of differential privacy. For details, see Section \ref{cdp}.\\
Definition \ref{securitygame} differs from the definition of Aggregator Obliviousness \cite{2} since we require the adversary to specify the set $U$ of non-compromised users \textit{before} making any query, i.e. we do not allow the adversary to determine $U$ adaptively.

\subsection{Feasibility of efficient and secure Private Stream Aggregation}

Using a secure PSA scheme, we ensure that the transmitted data of non-compromised users do not disclose sensitive information other than their aggregate. 
We now state the condition for the existence of secure (as in Definition \ref{securitygame}) PSA schemes for sum queries.

\begin{Thm}\label{PSATHEOREM}
Let $\kappa$ be a security parameter, and $m,n\in\mathbb{N}$ with $\log(m)=\text{poly}(\kappa),n=\text{poly}(\kappa)$. Let $(G,\cdot), (S,*)$ be finite groups and $G^\prime\subseteq G$. For some finite set $M$, let \[\mathcal{F}=\{\text{\upshape\sffamily F}_s\,|\,\text{\upshape\sffamily F}_s:M\to G^\prime\}_{s\in S}\] be a family of functions which are homomorphic over $S$ and \[\varphi:\{-mn,\ldots,mn\}\to G\] an $mn$-isomorphic embedding. If $\mathcal{F}$ is a weak PRF family, then the following PSA scheme $\Sigma=(\mbox{\upshape\sffamily Setup}, \mbox{\upshape\sffamily PSAEnc}, \mbox{\upshape\sffamily PSADec})$ is secure:
\begin{description}
\item \textbf{\mbox{\upshape \sffamily Setup}}: $(\mbox{\upshape\sffamily pp},T,s,s_1,\ldots,s_n)\leftarrow \mbox{\upshape\sffamily Setup}(1^\kappa)$, where $\mbox{\upshape\sffamily pp}$ are parameters of $G,G^\prime,S,M,\mathcal{F},\varphi$. The keys are $s_i\in_R S$ for all $i\in[n]$ with $s=(\bigast_{i=1}^n s_i)^{-1}$ and $T\subset M$ such that all $t\in T$ are chosen uniformly at random from $M$.
\item \textbf{\mbox{\upshape \sffamily PSAEnc}}: Compute $c_{i,t}=\text{\upshape\sffamily F}_{s_i}(t)\cdot \varphi(x_i)$ in $G$, where $x_i\in\widehat{\mathcal{D}}=\{-m,\ldots,m\}, t\in T$.
\item \textbf{\mbox{\upshape \sffamily PSADec}}: Compute $\varphi(\sum_{i=1}^n x_i)=\text{\upshape\sffamily F}_s(t)\cdot c_{1,t}\cdot\ldots\cdot c_{n,t}$ and invert.
\end{description}
\end{Thm}

Thus, we need a key-homomorphic weak PRF and a mapping which homomorphically aggregates all users' data. Since every data value is at most $m$, the scheme correctly retrieves the aggregate, which is at most $m\cdot n$, by the $mn$-isomorphic property of $\varphi$. Importantly, the product of all pseudo-random values $\text{\upshape\sffamily F}_s(t),\text{\upshape\sffamily F}_{s_1}(t),\ldots,\text{\upshape\sffamily F}_{s_n}(t)$ is the neutral element in the group $G$ for all $t\in T$. Since the values in $T$ are uniformly distributed in $M$, it is enough to require that $\mathcal{F}$ is a \textit{weak} PRF family. Thus, the statement of this theorem does not require a random oracle.\\
The proof of Theorem \ref{PSATHEOREM} works with a sequence of games and builds on the ideas of \cite{2}. The details of the proof are given in Appendix \ref{ptproof}. 
Here we just give the main ideas. In the first step (Lemma \ref{gameonetwo}) 
we construct a real-or-random version of the PSA security game, where encryption is performed equally likely with a weak PRF or a random function. We show that winning the PSA security game is at least as hard as winning its real-or-random version. The second step (Lemma \ref{gametwothree}) 
shows that the plaintext dependence of the ciphertexts generated in the game can be abolished. Since we are dealing with a non-adaptive security definition there is no need of simulating 
the random choice of time-steps by programming a random oracle as it is required in the proof of Aggregator Obliviousness by Shi et al. \cite{2}. Therefore, in contrast to \cite{2}, our result does not rely on a random oracle and the full proof works in the standard model. In the last step, (Lemma \ref{gamethreeprf})
using the hybrid argument, we show that winning the game is at least as hard as distinguishing the weak PRF from a random function. Here the adversary's specification of $U$ before making the first query allows the PRF distinguisher to be consistent with real random values or pseudo-random values in its replies to the queries. All in all, we get that winning the first game with our construction is at least as hard as distinguishing the weak PRF from a random function, completing the proof of Theorem \ref{PSATHEOREM}.

\subsection{An efficient Private Stream Aggregation scheme}

We give an instantiation of a secure PSA scheme consisting of efficient algorithms. Its security is based on the Decisional Diffie-Hellman (DDH) problem.

\begin{Exm}\label{DDHEXM} Let $q>m\cdot n$ and $p=2\cdot q+1$ be large primes. Let furthermore $G=\mathbb{Z}_{p^2}^*, S=\mathbb{Z}_{pq}, M=G^\prime=\mathcal{QR}_{p^2}$ and $g\in\mathbb{Z}_{p^2}^*$ with $\text{ord}(g)=pq$. Then $g$ generates the group $M=G^\prime=\mathcal{QR}_{p^2}$ of quadratic residues modulo $p^2$. In this group DDH is assumed to be hard. Then we define 
\begin{itemize}
\item $\{\mbox{\upshape\sffamily pp}\}=(g,p)$. Choose keys $s_1,\ldots,s_n\in_R\mathbb{Z}_{pq}$ and $s=-\sum_{i=1}^n s_i$ mod $pq$. Let $T\subset M$, i.e. $t$ is a power of $g$ for every $t\in T$.
\item $\text{\upshape\sffamily F}_{s_i}(t)=t^{s_i}\mbox{ mod } p^2$. This is a weak PRF under the DDH assumption, as can be proven using arguments similar to the ones in $\cite{18}$. 
\item $\varphi(x_i)=1+p\cdot x_i$ mod $p^2$, where $-m\leq x_i\leq m$. (It is easy to see that $\varphi$ is an $mn$-isomorphic embedding.)
\end{itemize}

\noindent For aggregation, we compute $X\in\{1-p\cdot mn,\ldots,1+p\cdot mn\}$ with
\begin{align*} X\equiv & \text{\upshape\sffamily F}_{s}(t)\cdot\prod_{i=1}^n \text{\upshape\sffamily F}_{s_i}(t)\cdot\varphi(x_i)\equiv\prod_{i=1}^n (1+p\cdot x_i)\\
 \equiv & 1+p\cdot\sum_{i=1}^n x_i+p^2\cdot\sum_{i,j\in[n],j\neq i} x_i x_j+\ldots+p^n\cdot\prod_{i=1}^n x_i\\
 \equiv & 1+p\cdot\sum_{i=1}^n x_i\mbox{ mod } p^2
\end{align*}
and decrypt $\sum_{i=1}^n x_i=\frac{1}{p}(X-1)$ over the integers.
\end{Exm}

The difference to the scheme introduced in \cite{2} lies on the map $\varphi$. Whereas the PRF in \cite{2} works similarly (the underlying group $G$ is $\mathbb{Z}_p^*$ rather than $\mathbb{Z}_{p^2}^*$), the aggregational function is defined by
\[\varphi(x_i)=g^{x_i}\mbox{ mod } p,\]
which requires to solve the discrete logarithm modulo $p$ for decrypting. In contrast, our efficient construction only requires a subtraction and a division over the integers.

\begin{Rem} In the random oracle model, the construction shown in Example \ref{DDHEXM} achieves the stronger notion of Aggregator Obliviousness, which is the adaptive version of our security definition (for details, see the proof in Appendix A of $\cite {2}$). The same proof can be applied to our instantiation by simply replacing the map $\varphi$ involved and using a strong version of the PRF $\text{\upshape\sffamily F}$.\footnote{For showing Aggregator Obliviousness we would have to substitute the choice of random values $t$ by a hash function $H:M\to\mathcal{QR}_{p^2}$ modeled as a random oracle for some domain $M$. Therefore the PRF would be $\text{\upshape\sffamily F}_{s_i}(t)=H(t)^{s_i}\mbox{ mod } p^2$ which is the strong version of the weak PRF in Example \ref{DDHEXM}. In this case, all $t$ may be chosen in a deterministic way.}
\end{Rem}

\begin{table}[t]
\centering
\begin{tabular}{l||c|c|c} Length of $p$ & $1024$-bit & $2048$-bit & $4096$-bit\\ \hline \cite{2} & $1.1$ ms & $7.5$ ms & $57.0$ ms\\ This work & $3.9$ ms & $29.4$ ms & $225.0$ ms
\end{tabular}
\captionof{table}{Time for encryption}\label{enctab}
\end{table}

\begin{table}[t]
\centering
\begin{tabular}{l||c|c|c|c|c} \,\,\,\,\,\,\,\,\,$m$ & $10^0$ & $10^1$ & $10^2$ & $10^3$ & $10^4$\\ \hline \cite{2} b.-f. & $0.04$ s & $0.24$ s & $2.67$ s & $28.97$ s & $381.05$ s\\ This work & $0.09$ s & $0.08$ s & $0.08$ s & $0.09$ s & $0.08$ s
\end{tabular}
\captionof{table}{Time for decryption ($2048$ bit, $n=1000$)}\label{dectab}
\end{table}

Note that for a given $p$, the running time of the decryption in our scheme does not depend on $m$, so it provides a small running time even if $m$ is exponentially large. At the same time, the decryption of the scheme in \cite{2} can also be performed efficiently, even if $m$ is superpolynomial in the security parameter: discretise the plaintext space into $\sqrt{n}$ equidistant values and let each user choose the value nearest to her original one as input for encryption. The aggregated value has the correct expectation, but the decryption algorithm has to search only in a range of $n^{1.5}$ values. However, this method causes a superlinear time-dependence on $n$ for the decryption and induces an additional aggregation error due to discretisation.\\
We compare the practical running times for encryption and decryption of the scheme in \cite{2} with the algorithms of our scheme in Table \ref{enctab} and Table \ref{dectab}, respectively. Here, let $m$ denote the size of the plaintext space. Encryption is compared at different security levels with $m=1$. For comparing the decryption time, we fix the security level and the number of users and let $m$ be variable.\\
All algorithms were executed on an Intel Core i$5$, $64$-bit CPU at $2.67$ GHz. We compared the schemes at the same security level, assuming that the DDH problem modulo $p$ is as hard as modulo $p^2$, i.e. we used the same value for $p$ in both schemes. For different bit-lengths of $p$, we observe that the encryption of our scheme is roughly $4$ times slower than the encryption in \cite{2}. The running time of our decryption algorithm is widely dominated by the aggregation phase. Therefore it is clear, that it linearly depends on $n$. 
Using a $2048$-bit prime and fixing $n=1000$, the running time of the decryption in our scheme is less than $0.1$ second for varying values of $m$. In contrast, the time for the brute-force decryption in \cite{2} grows roughly linearly in $m$.\\
As observed in \cite{2}, using Pollard's lambda method would reduce the running time for decryption of \cite{2} to about $\sqrt{mn}$. Nevertheless, our scheme provides a speed-up of $\sqrt{m/n}$ whenever $m$ is larger than $n$, while the encryption is decelerated only by a constant factor.

\section{Achieving Computational Differential Privacy}\label{cdp}

\begin{Ntn} Let $\kappa$ be a security parameter. If an expression $\omega=\omega(\kappa)$ is non-negligible in $\kappa$ (i.e. if $\omega> 1/\text{poly}(\kappa)$), then we write $\omega>\text{\upshape\sffamily neg}(\kappa)$.
\end{Ntn}

In this section, we describe how to preserve computational differential privacy using a PSA scheme. As described above, in the work by Chan et al. \cite{3} the polynomial-time reduction between an attacker against the security of a PSA scheme and an attacker against differential privacy is missing. In this section we provide an appropriate reduction. The content of this section is independent of Theorem \ref{PSATHEOREM}. Specifically, let $\mathcal{A}$ be a mechanism which, given some event $Good$, evaluates a statistical query $f:\mathcal{D}^n\to\mathcal{O}$ over a database $D\in\mathcal{D}^n$ preserving $\epsilon$-\mbox{\upshape\sffamily DP}. Furthermore, let $\Sigma$ be a secure PSA scheme for $f$. We show that $\mathcal{A}$ executed through $\Sigma$ preserves $\epsilon$-\mbox{\upshape\sffamily CDP} given $Good$. Let $Bad=\overline{Good}$ and assume $\Pr[Bad]\leq\delta$. In Section \ref{dpmech}, we will give instantiations of such a mechanism $\mathcal{A}$ 
and show that they preserve $(\epsilon,\delta)$-\mbox{\upshape\sffamily CDP} \textit{unconditionally} if executed through $\Sigma$. For simplicity, in this section we focus on sum queries, but our analysis can be easily extended to more general statistical queries. Our technique involves a reduction-based proof using a biased coin toss and is of independent interest.

\subsection{Redefining the security of Private Stream Aggregation}
Let us first modify the security game in Definition \ref{securitygame} in the following way. Let \textbf{game} $\boldsymbol 1$ be the original game from Definition \ref{securitygame}. Let $p\in(0,1)$ and $P=\max\{p,1-p\}$. The $\boldsymbol P$-\textbf{game} $\boldsymbol 1$ for a probabilistic polynomial-time adversary $\mathcal{T}_{1,P}$ is defined as \textbf{game} $\boldsymbol 1$ with the following changes:
\begin{itemize}
 \item Before the challenge phase, $\mathcal{T}_{1,P}$ sends $p$ to the challenger.
 \item In the challenge phase, the challenger chooses $b=0$ with probability $p$ and $b=1$ with probability $1-p$.
\end{itemize}
We call a PSA scheme $P$-secure if the probability of every probabilistic polynomial-time adversary $\mathcal{T}_{1,P}$ in winning the above game is $P+\mbox{\upshape\sffamily neg}(\kappa)$. Note that \textbf{game} $\boldsymbol 1$ is a special case of $\boldsymbol P$-\textbf{game} $\boldsymbol 1$, where $P=1/2$. We refer to this case as the unbiased version (rather the biased version if $P>1/2$) of $\boldsymbol P$-\textbf{game} $\boldsymbol 1$. In the unbiased case, we just drop the dependence on $P$ and the adversary is not required to send $p$ to its challenger.

\subsection{Constructing a PSA adversary using a CDP adversary}

\subsubsection{Security game for adjacent databases}
For showing that a $P$-secure PSA scheme is suitable for preserving {\upshape\sffamily CDP}, we have to construct a successful adversary in $\boldsymbol P$-\textbf{game} $\boldsymbol 1$ (with a proper choice of $p$) using a successful distinguisher for adjacent databases. We define the following \textbf{game} $\boldsymbol 0$ for a probabilistic polynomial-time adversary $\mathcal{T}_0$ which is identical to \textbf{game} $\boldsymbol 1$ with a changed challenge-phase:
\begin{description}
\item\textbf{Challenge.} $\mathcal{T}_0$ chooses $t^*\in T$ such that no encryption query at $t^*$ was made. $\mathcal{T}_0$ queries two \textit{adjacent} tuples $(d_i^{[0]})_{i\in U},(d_i^{[1]})_{i\in U}$. For all $i\in U$ the challenger returns 
\[c_{i,t^*}=\mbox{\upshape\sffamily PSAEnc}_{s_i}(t^*,x_i^{[b]}),\]
where $x_i^{[b]}\in\widehat{\mathcal{D}}$ is a noisy version of $d_i^{[b]}\in\mathcal{D}$ for all $i\in U$ obtained by some randomised perturbation process.
\end{description}
Now consider the following experiment which we call $\mbox{\upshape \sffamily Exp}_1$. Let $f_{\mathcal{D}}:\mathcal{D}^n\to\mathcal{O}$ be a sum query and $Q:\widehat{\mathcal{D}}^n\to[0,1]$ a probability distribution function on $\widehat{\mathcal{D}}^n$. For simplicity we consider only the case where $\mathcal{D}\subseteq\widehat{\mathcal{D}}=\mathbb{Z}$. Let $D_0,D_1\in\mathcal{D}^n$. Then $\mbox{\upshape \sffamily Exp}_1$ is performed as follows:
\begin{itemize} 
\item Let $B_{1/2}$ be a Bernoulli variable with $\Pr[B_{1/2}=0]=1/2$.
\item Let $\widehat{X}_1$ be a random vector with probability distribution function $Q(x-D_b)$, where $b$ is a realisation of $B_{1/2}$.
\item Let $Y=f_{\widehat{\mathcal{D}}}(\widehat{X}_1)$.
\end{itemize}
We now define an experiment $\mbox{\upshape \sffamily Exp}_2$ and afterwards show that it is statistically equivalent to $\mbox{\upshape \sffamily Exp}_1$.
\begin{itemize} 
\item Let $Y$ be a random variable as in $\mbox{\upshape \sffamily Exp}_1$.
\item Let $p=\Pr[B_{1/2}=0|Y=y]$. Let $B_{p}$ be a Bernoulli variable with $\Pr[B_{p}=0]=p$.
\item Let $\widehat{X}_2$ be a random vector with conditional probability function
\begin{align*} 
   & \Pr[\widehat{X}_2=\widehat{D}| B_{p}=b, Y=y]\\
 = & \begin{cases} \frac{\chi(y)\cdot Q(\widehat{D}-D_0)}{2\cdot p\cdot\Pr[Y=y]}, &\mbox{ if } b=0\\ \frac{\chi(y)\cdot Q(\widehat{D}-D_1)}{2\cdot (1-p)\cdot\Pr[Y=y]}, &\mbox{ if } b=1.\end{cases}
\end{align*}
Here $b$ is a realisation of $B_p$ and $\chi(y)=\chi_{\{z\in\mathcal{O} | z=f_{\widehat{\mathcal{D}}}(\widehat{D})\}}(y)$ denotes the characteristic function of $\{z\in\mathcal{O} | z=f_{\widehat{\mathcal{D}}}(\widehat{D})\}\subset\mathcal{O}$, which is $1$ if $y=f_{\widehat{\mathcal{D}}}(\widehat{D})$ and $0$ otherwise. Note that the values $\chi(y),Q(x),p, \Pr[Y=y]$ for the computation of the conditional probability of $\widehat{X}_2$ are known.
\end{itemize}

For showing that $\mbox{\upshape \sffamily Exp}_1$ and $\mbox{\upshape \sffamily Exp}_2$ are statistically equivalent it suffices to show that the joint distributions $\Pr[B_{1/2},\widehat{X}_1,Y]$ and $\Pr[B_{p},\widehat{X}_2,Y]$ are equal.

\begin{Lem}\label{EXPEQUIV} $\Pr[B_{1/2},\widehat{X}_1,Y]=\Pr[B_{p},\widehat{X}_2,Y]$.
\end{Lem}
\begin{proof}
We observe that in $\mbox{\upshape \sffamily Exp}_1$, $\Pr[Y=y| \widehat{X}_1=\widehat{D}, B_{1/2}=b]=\Pr[Y=y| \widehat{X}_1=\widehat{D}]$. Therefore we have
\begin{align*} 
& \Pr[\widehat{X}_1=\widehat{D}| B_{1/2}=b, Y=y]\\
 = & \frac{\Pr[Y=y| \widehat{X}_1=\widehat{D}]\cdot Q(\widehat{D}-D_b)}{\Pr[Y=y| B_{1/2}=b]}\\
 = & \frac{\chi(y)\cdot Q(\widehat{D}-D_b)}{\Pr[Y=y| B_{1/2}=b]},
\end{align*}
which exactly corresponds to the conditional probability of $\widehat{X}_2$ in $\mbox{\upshape \sffamily Exp}_2$. Thus, we have
\begin{align*}
& \Pr[B_{p},\widehat{X}_2,Y]\\
= & \Pr[Y]\cdot\Pr[B_p]\cdot\Pr[\widehat{X}_2| B_{p},Y]\\
= & \Pr[Y]\cdot\Pr[B_{1/2}| Y]\cdot\Pr[\widehat{X}_1| B_{1/2},Y]\\
= & \Pr[B_{1/2},\widehat{X}_1,Y]. 
\end{align*}
\end{proof}

\noindent 
Note that Lemma \ref{EXPEQUIV} also applies to the marginals of the triples $(B_{1/2},\widehat{X}_1,Y)$ and $(B_{p},\widehat{X}_2,Y)$.

\subsubsection{The Reduction}

With Lemma \ref{EXPEQUIV} in mind, we can show that a successful adversary in \textbf{game} $\boldsymbol 0$ yields a successful adversary in $\boldsymbol P$-\textbf{game} $\boldsymbol 1$ for a particular $P\in[1/2,1)$. Afterwards we show that a successful adversary in $\boldsymbol P$-\textbf{game} $\boldsymbol 1$ for \textit{any} $P\in[1/2,1)$ yields a successful adversary in \textbf{game} $\boldsymbol 1$.

\begin{Lem}\label{gamezeroPone} Let $\kappa$ be a security parameter. Let $\mathcal{T}_0$ be an adversary in \textbf{\upshape game} $\boldsymbol 0$ with advantage $\mu_{0}(\kappa)>\text{\upshape\sffamily neg}(\kappa)$. Let $B_{1/2}$ denote the random variable describing the challenge bit $b$ in \textbf{\upshape game} $\boldsymbol 0$ and let $Y$ denote the random variable describing the aggregate of $(x_i^{[b]})_{i\in U}$. Let $p$ be the probability of $B_{1/2}=0$ given the choice of $Y$ and let $P=\max\{p,1-p\}$. Then there exists an adversary $\mathcal{T}_{1,P}$ in $\boldsymbol P$-\textbf{\upshape game} $\boldsymbol 1$ with advantage $\mu_{1,P}(\kappa)>\text{\upshape\sffamily neg}(\kappa)$.
\end{Lem}
\begin{proof} We construct a successful adversary $\mathcal{T}_{1,P}$ in $\boldsymbol P$-\textbf{game} $\boldsymbol 1$ using $\mathcal{T}_0$ as follows:
\noindent\begin{description}
 \item\textbf{Setup.} Receive $\kappa, \mbox{\upshape\sffamily pp},T,s$ from the $\boldsymbol P$-\textbf{game} $\boldsymbol 1$-challenger and send it to $\mathcal{T}_0$.
\item\textbf{Queries.} Receive $U=\{i_1,\ldots,i_u\}\subseteq[n]$ from $\mathcal{T}_0$ and send it to the challenger. Forward the obtained response $(s_i)_{i\in[n]\setminus U}$ to $\mathcal{T}_0$. Forward $\mathcal{T}_0$'s queries $(i,t,d_i)$ with $i\in U, t\in T, d_i\in\mathcal{D}$ to the challenger and forward the obtained response $c_{i,t}$ to $\mathcal{T}_0$.
\item\textbf{Challenge.} $\mathcal{T}_0$ chooses $t^*\in T$ such that no encryption query at $t^*$ was made and queries two adjacent tuples $(d_i^{[0]})_{i\in U},(d_i^{[1]})_{i\in U}$. Choose a realisation $y$ of $Y$ according to $\mbox{\upshape \sffamily Exp}_2$. Set $p=\Pr[B_{1/2}=0| Y=y]$ and choose $(x_i^{[a]})_{i\in U}$ with probability $\Pr[\widehat{X}_2=(x_i^{[a]})_{i\in U}| B_{1/2}=a, Y=y]$ for $a=0,1$ according to $\mbox{\upshape \sffamily Exp}_2$. Send $p,t^*,\linebreak(x_i^{[0]},x_i^{[1]})_{i\in U}$ to the challenger. Obtain the response $(c_{i,t^*})_{i\in U}$ and forward it to $\mathcal{T}_0$. 
\item\textbf{Queries.} $\mathcal{T}_0$ can make the same type of queries as before with the restriction that no encryption query at $t^*$ can be made.
\item\textbf{Guess.} $\mathcal{T}_0$ gives a guess about which database was encrypted. Output the same guess.
\end{description}
The rules of $\boldsymbol P$-\textbf{game} $\boldsymbol 1$ are preserved since $\mathcal{T}_{1,P}$ sends two tuples of the same aggregate $y$ to its challenger. On the other hand, since the ciphertexts generated by the challenger are determined by the challenge bit and the collection $(x_i^{[b]})_{i\in U}$, the rules of \textbf{game} $\boldsymbol 0$ are preserved by Lemma~\ref{EXPEQUIV} (the triple $(b,(x_i^{[b]})_{i\in U},y)$ is chosen according to $\mbox{\upshape \sffamily Exp}_2$). Therefore $\mathcal{T}_{1,P}$ perfectly simulates \textbf{game} $\boldsymbol 0$ and has the same advantage as $\mathcal{T}_0$.
\end{proof}

We now show that a secure PSA scheme is also $P$-secure for every $p\in(0,1)$, where $P=\max\{p,1-p\}$.

\begin{Lem}\label{gamePoneone} Let $\kappa$ be a security parameter. For any $p\in(0,1)$ let $\mathcal{T}_{1,P}$ be an adversary in $\boldsymbol P$-\textbf{\upshape game $\boldsymbol 1$} with advantage $\mu_{1,P}(\kappa)>\text{\upshape\sffamily neg}(\kappa)$. Then there exists an adversary $\mathcal{T}_1$ in \textbf{\upshape game $\boldsymbol 1$} with advantage $\mu_{1}(\kappa)>\text{\upshape\sffamily neg}(\kappa)$.
\end{Lem}
\begin{proof} Given a successful adversary $\mathcal{T}_{1,P}$ in $\boldsymbol P$-\textbf{game} $\boldsymbol 1$, we construct a successful adversary $\mathcal{T}_1$ in \textbf{game} $\boldsymbol 1$ as follows:
\noindent\begin{description}
 \item\textbf{Setup.} Receive $\kappa, \mbox{\upshape\sffamily pp},T,s$ from the \textbf{game} $\boldsymbol 1$-challenger and send it to $\mathcal{T}_{1,P}$.
\item\textbf{Queries.} Receive $U\subseteq[n]$ from $\mathcal{T}_{1,P}$ and send it to the challenger. Forward the obtained response $(s_i)_{i\in[n]\setminus U}$ to $\mathcal{T}_{1,P}$. Forward $\mathcal{T}_{1,P}$'s queries $(i,t,x_i)$ with $i\in U, t\in T, x_i\in\widehat{\mathcal{D}}$ to the challenger and forward the obtained response $c_{i,t}$ to $\mathcal{T}_{1,P}$.
\item\textbf{Challenge.} $\mathcal{T}_{1,P}$ chooses $t^*\in T$ such that no encryption query at $t^*$ was made, sends $p\in(0,1)$ and queries two different tuples $(x_i^{[0]})_{i\in U},(x_i^{[1]})_{i\in U}$ with $f_{\widehat{\mathcal{D}}|_U}((x_i^{[0]})_{i\in U})\linebreak =f_{\widehat{\mathcal{D}}|_U}((x_i^{[1]})_{i\in U})$. Choose a bit $a$ with $\Pr[a=0]=p, \Pr[a=1]=1-p$ and query $(x_i^{[a]})_{i\in U},(x_i)_{i\in U}$ to the challenger, where the $x_i$ are chosen uniformly at random from $\widehat{\mathcal{D}}$ such that $f_{\widehat{\mathcal{D}}|_U}((x_i)_{i\in U})=f_{\widehat{\mathcal{D}}|_U}((x_i^{[a]})_{i\in U})$. Obtain the response $(c_{i,t^*})_{i\in U}$ and forward it to $\mathcal{T}_{1,P}$. 
\item\textbf{Queries.} $\mathcal{T}_{1,P}$ can make the same type of queries as before with the restriction that no encryption query at $t^*$ can be made.
\item\textbf{Guess.} $\mathcal{T}_{1,P}$ gives a guess about $a$. If the guess is correct, then output $0$; if not, output $1$.
\end{description}
If $\mathcal{T}_{1,P}$ has output the correct guess about $a$ then $\mathcal{T}_1$ can say with high confidence that the challenge ciphertexts are the encryptions of $(x_i^{[a]})_{i\in U}$ and therefore outputs $0$. On the other hand, if $\mathcal{T}_{1,P}$'s guess was not correct, then $\mathcal{T}_1$ can say with high confidence that the challenge ciphertexts are the encryptions of the random collection $(x_i)_{i\in U}$ and it outputs $1$. Formally:\par\medskip

\noindent\textbf{Case $\boldsymbol 1$.} Let $(c_{i,t^*})_{i\in U}=(\text{\upshape\sffamily PSAEnc}_{s_i}(x_i^{[a]}))_{i\in U}$. Then $\mathcal{T}_1$ perfectly simulates $\boldsymbol P$-\textbf{game} $\boldsymbol 1$ for $\mathcal{T}_{1,P}$ and the distribution of ciphertexts is the same as in $\boldsymbol P$-\textbf{game} $\boldsymbol 1$:
\begin{align*} & \Pr[\mathcal{T}_1 \text{ outputs } 0]\\
 = & p\cdot \Pr[\mathcal{T}_{1,P} \text{ outputs } 0 \,|\, a=0]+(1-p)\cdot\Pr[\mathcal{T}_{1,P} \text{ outputs } 1 \,|\, a=1]\\
 = & \Pr[\mathcal{T}_{1,P} \text{ wins } \boldsymbol P\text{-\textbf{game} } \boldsymbol 1]\\
 = & P + \mu_{1,P}(\kappa).
\end{align*}\par\medskip

\noindent\textbf{Case $\boldsymbol 2$.} Let $(c_{i,t^*})_{i\in U}=(\text{\upshape\sffamily PSAEnc}_{s_i}(x_i))_{i\in U}$. Then the ciphertexts are random with the constraint 
that their product is the same as in the first case. The probability that $\mathcal{T}_{1,P}$ wins \textbf{game} $\boldsymbol 1$ is at most $P$ and
\begin{align*} & \Pr[\mathcal{T}_1 \text{ outputs } 1]\\
 = & p\cdot\Pr[\mathcal{T}_{1,P} \text{ outputs } 1 \,|\, a=0]+(1-p)\cdot\Pr[\mathcal{T}_{1,P} \text{ outputs } 0 \,|\, a=1]\cdot (1-p)\\
 = & \Pr[\mathcal{T}_{1,P} \text{ loses } \boldsymbol P\text{-\textbf{game} } \boldsymbol 1]\\
 \geq & 1-P.
\end{align*}
Finally we obtain that the advantage of $\mathcal{T}_1$ in winning \textbf{game} $\boldsymbol 1$ is
\[\mu_{1}(\kappa)\geq\frac{1}{2}\mu_{1,P}(\kappa)>\text{\upshape\sffamily neg}(\kappa).\]
\end{proof}

\subsection{Proof of Computational Differential Privacy}

We have shown that no probabilistic polynomial-time adversary can win \textbf{game} $\boldsymbol 0$ if the underlying PSA scheme is secure. If the perturbation process in \textbf{game} $\boldsymbol 0$ is $\epsilon$-{\sffamily DP}-preserv\-ing, then the whole construction provides $\epsilon$-{\sffamily CDP}, as we show now.

\begin{Thm}\label{cdptheorem} Let $\mathcal{A}$ be a mechanism for a query $f_{\mathcal{D}}:\mathcal{D}^n\to\mathcal{O}$ which preserves $\epsilon$-\mbox{\upshape\sffamily DP} and let $\Sigma$ be a secure PSA scheme for $f_{\mathcal{D}}$. Then $\mathcal{A}$ preserves $\epsilon$-\mbox{\upshape\sffamily CDP} if it is used for the perturbation process in \textbf{\upshape game $\boldsymbol 0$} instantiated with $\Sigma$.
\end{Thm}
\begin{proof} Consider again \textbf{game} $\boldsymbol 1$, $\boldsymbol P$-\textbf{game} $\boldsymbol 1$ and \textbf{game} $\boldsymbol 0$. We first bound the probability $p=\Pr[B_{1/2}=0| Y=y]$ for the biased coin in $\boldsymbol P$-\textbf{game} $\boldsymbol 1$. Since the perturbation process was performed by $\mathcal{A}$, the random variable $Y$ corresponds to the output of $\mathcal{A}$ and we have
\begin{align*} e^{-\epsilon}\cdot\Pr[Y=y| B_{1/2}=0]\leq & \phantom{e^{-\epsilon}\cdot}\,\Pr[Y=y| B_{1/2}=1]\\
\leq & e^\epsilon\,\,\,\cdot\Pr[Y=y| B_{1/2}=0].
\end{align*}
By the Bayes-formula we get
\[p_{min}:=\frac{1}{1+e^\epsilon}\leq p\leq\frac{e^\epsilon}{1+e^\epsilon}=:p_{max}.\]
Now let $\mathcal{T}$ be a probabilistic polynomial-time Turing machine. Let $\mathcal{T}_{1,P}$ denote this Turing machine as adversary in the $\boldsymbol P$-\textbf{game} $\boldsymbol 1$ for any $P=\max\{p,1-p\}$ with\linebreak $p\in[p_{min},p_{max}]$ and let $\mathcal{T}_1$ denote the same Turing machine as adversary in \textbf{game} $\boldsymbol 1$. Let finally $\mathcal{D}_{\mbox{\scriptsize\upshape\sffamily CDP}}=\mathcal{T}_0$ denote the same machine as adversary in \textbf{game} $\boldsymbol 0$. Then for $a=0,1$:
\begin{align*} 
 \Pr[\mathcal{D}_{\mbox{\scriptsize\upshape\sffamily CDP}}=a, B_{1/2}=1]= & \phantom{p_{max}\cdot}\,\,\,\Pr[\mathcal{T}_{1,P}=a, B_p=1]\numberthis\label{eqcdp1}\\
 \leq & p_{max}\,\cdot\Pr[\mathcal{T}_{1,P}=a| B_p=1]\\
 = & p_{max}\,\cdot\Pr[\mathcal{T}_1=a| B_{1/2}=1]\numberthis\label{eqone1}\\
 \leq & p_{max}\,\cdot\Pr[\mathcal{T}_1=a| B_{1/2}=0]+\mbox{\upshape\sffamily neg}(\kappa)\\
 = & p_{max}\,\cdot\Pr[\mathcal{T}_{1,P}=a| B_p=0]+\mbox{\upshape\sffamily neg}(\kappa)\numberthis\label{eqone2}\\
 \leq & \frac{p_{max}}{p_{min}}\cdot\Pr[\mathcal{T}_{1,P}=a, B_p=0]+\mbox{\upshape\sffamily neg}(\kappa)\\
 = & \,\,\,\,\,\,\,e^\epsilon\,\,\cdot\Pr[\mathcal{T}_{1,P}=a, B_p=0]+\mbox{\upshape\sffamily neg}(\kappa)\\
 = & \,\,\,\,\,\,\,e^\epsilon\,\,\cdot\Pr[\mathcal{D}_{\mbox{\scriptsize\upshape\sffamily CDP}}=a, B_{1/2}=0]+\mbox{\upshape\sffamily neg}(\kappa).\numberthis\label{eqcdp2}
\end{align*}
Equations \eqref{eqcdp1} and \eqref{eqcdp2} hold because of Lemma \ref{gamezeroPone} and Equations \eqref{eqone1} and \eqref{eqone2} hold because of Lemma \ref{gamePoneone}. 
It follows that
\[\Pr[\mathcal{D}_{\mbox{\scriptsize\upshape\sffamily CDP}}=a| B_{1/2}=1]\leq e^\epsilon\cdot \Pr[\mathcal{D}_{\mbox{\scriptsize\upshape\sffamily CDP}}=a| B_{1/2}=0]+\mbox{\upshape\sffamily neg}(\kappa).\]
\end{proof}

As mentioned at the beginning of this section, we are considering a mechanism which preserves $\epsilon$-\mbox{\upshape\sffamily DP} given some event $Good$. Therefore, also Theorem \ref{cdptheorem} applies to this mechanism given $Good$. Accordingly, the mechanism \textit{unconditionally} preserves $(\epsilon,\delta)$-\mbox{\upshape\sffamily CDP}, where $\delta$ is a bound on the probability that $Good$ does not occur.

\section{Mechanisms for Differential Privacy}\label{dpmech}

\begin{figure*}\centering
\includegraphics[scale=0.4]{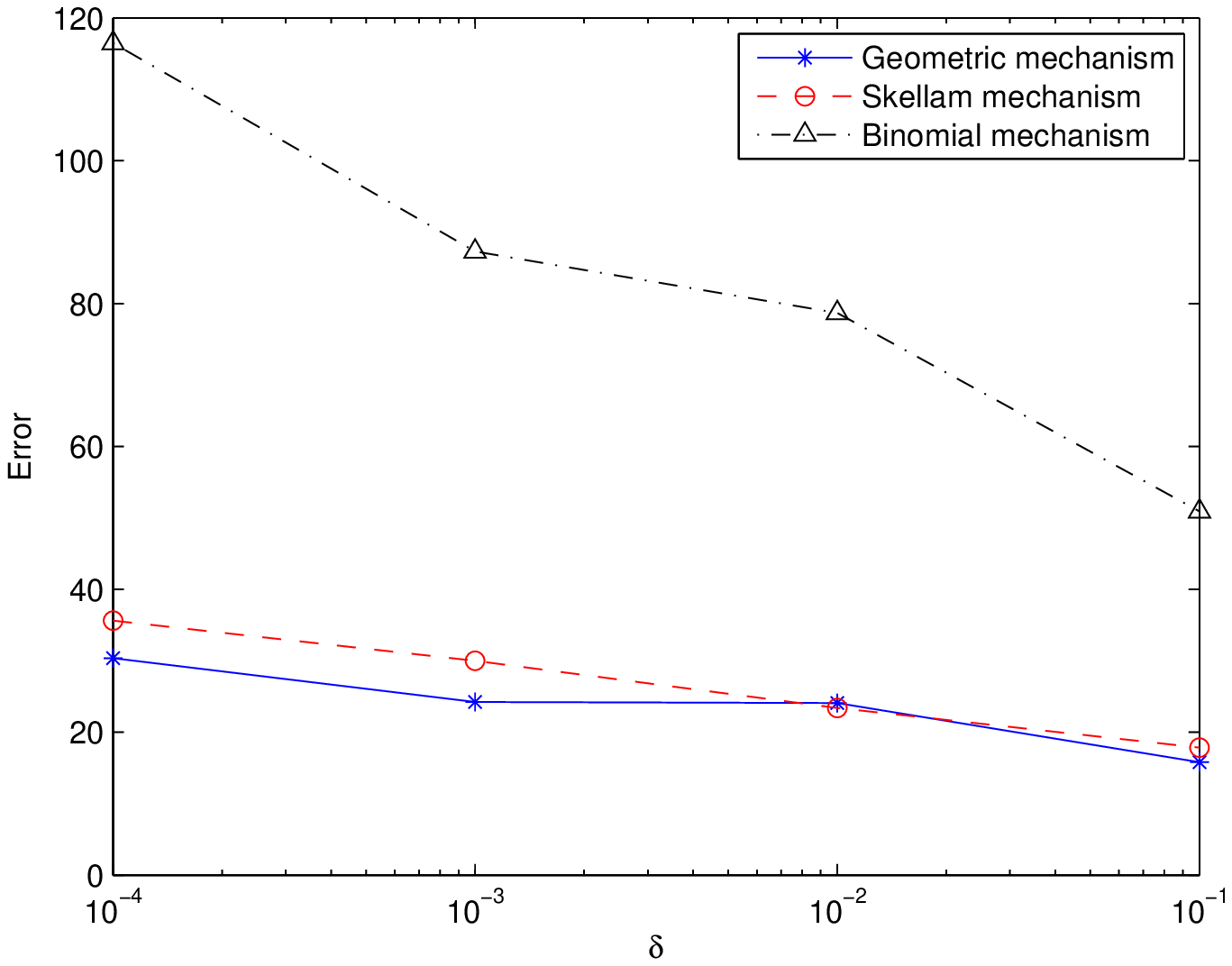}\,\,\,\,\,\,\,\,\,\,\,\,\,\,\,\,\,\,\,\,\,\includegraphics[scale=0.4]{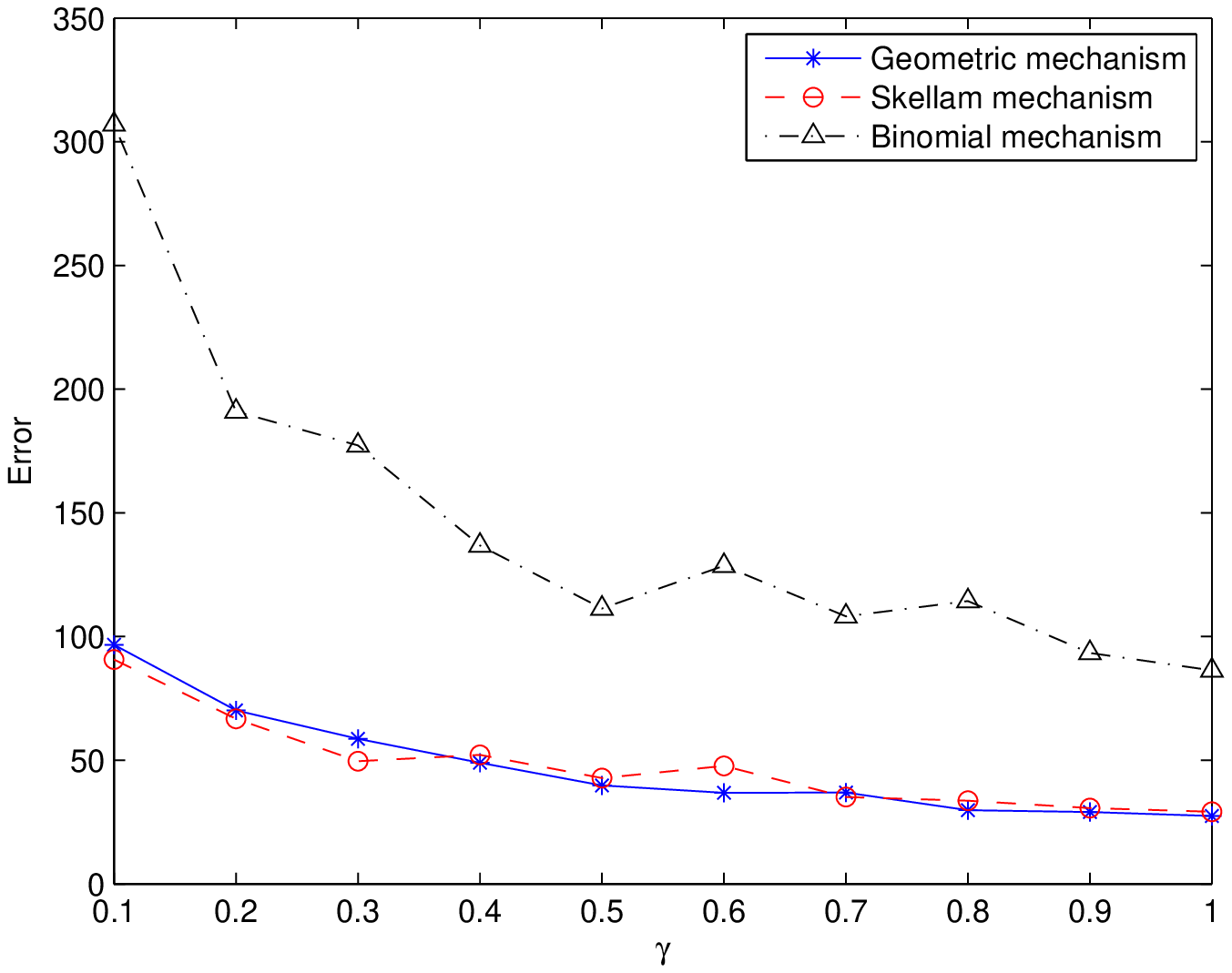}
\caption{Empirical error of the Geometric, Skellam and Binomial mechanisms. The fixed parameters are $\epsilon=0.1, S(f)=1, \beta=0.001$. The left graph shows the mean of the error in absolute value for variable $\delta$ and $\gamma=1$ over $100$ runs, the right graph is for variable $\gamma$ and $\delta=0.001$.}\label{accuracycomp}
\end{figure*}

In this section we recall the Geometric mechanism from \cite{2} and the Binomial mechanism from \cite{14} and introduce the Skellam mechanism. Since these mechanisms make use of a discrete probability distribution, they are well-suited for an execution through a secure PSA scheme, thereby preserving computational differential privacy as shown in the last section.

\begin{Thm}\label{privthm} Let $\epsilon>0$. For all databases $D\in\mathcal{D}^n$ the randomised mechanism \[\mathcal{A}(D):=f(D)+Y\]
preserves $(\epsilon,\delta)$-\mbox{\upshape\sffamily DP} with respect to any query $f$ with sensitivity $S(f)$, if $Y$ is distributed according to one of the following probability distributions:
\begin{enumerate}
\item $Y\sim\text{\upshape Geom}(\lambda)$ with $\lambda=\exp(\epsilon/S(f))$ (and $\delta=0$) \text{\upshape \cite{2}},
\item $Y\sim \text{\upshape Bin}(n',1/2)$ with $n'=64\cdot S(f)^2\cdot\log(2/\delta)/\epsilon^2$ \text{\upshape \cite{14}},
\item $Y\sim\text{\upshape Sk}(\mu)$\footnote{$\text{Sk}(\mu)$ denotes the symmetric Skellam distribution with mean $0$ and variance $\mu$. For details, see Appendix \ref{skellamsec}.} with 
\[\mu=\frac{\log(1/\delta)}{1-\cosh(\epsilon/S(f))+(\epsilon/S(f))\cdot\sinh(\epsilon/S(f))}.\]
\end{enumerate}
\end{Thm}

We provide the proof of the third claim in Appendix \ref{skellamsec}. Executing these mechanisms through a PSA scheme requires the use of the known constant $\gamma$ which denotes the a priori estimate of the lower bound on the fraction of non-compromised users. For this case, we provide the accuracy bounds for the aforementioned mechanisms.

\begin{Thm}\label{errorthm} Let $\epsilon>0, 0<\delta<1, S(f)>0$ and let $0<\gamma<1$ be the a priori estimate of the lower bound on the fraction of non-compromised users in the network. By distributing the execution of a perturbation mechanism as described above and using the parameters from Theorem \ref{privthm}, we obtain $(\alpha,\beta)$-accuracy with the following parameters:
\begin{enumerate}
\item $\alpha=\frac{4\cdot S(f)}{\epsilon}\cdot\sqrt{\frac{1}{\gamma}\cdot\log\left(\frac{1}{\delta}\right)\cdot\log\left(\frac{2}{\beta}\right)}$ for the Geometric mechanism, where $\delta$ bounds the probability that no user has added noise \text{\upshape \cite{2}},
\item $\alpha=\frac{8\sqrt{2}\cdot S(f)}{\epsilon}\cdot\sqrt{\frac{1}{\gamma}\cdot\log\left(\frac{2}{\delta}\right)\cdot\log\left(\frac{2}{\beta}\right)}$ for the Binomial mechanism,
\item $\alpha=\frac{S(f)}{\epsilon}\cdot\left(\frac{1}{\gamma}\cdot\log\left(\frac{1}{\delta}\right)+\log\left(\frac{2}{\beta}\right)\right)$ for the Skellam mechanism.
\end{enumerate}
\end{Thm}

The second claim can be easily shown using a standard tail bound for the Binomial distribution. The proof of the third claim is provided in Appendix \ref{skellamsec}.\\
Theorem \ref{errorthm} shows that for constant $\delta,\beta,\gamma$ the errors of the three mechanisms are bounded by $O(S(f)/\epsilon)$ and therefore do not exceed known bounds in the centralised model. As pointed out in Section \ref{mechov}, the execution of the Geometric mechanism through a PSA scheme requires each user to generate full noise with a small probability. Complementary, the other two mechanisms allow all users to simply generate noise of small variance. While the accuracy bound of the Geometric is roughly a constant factor smaller than the bound of the Binomial, we obtain a better bound for this second approach using the Skellam mechanism. Specifically, the ratio between the factor $\log(2/\beta)+\log(1/\delta)/\gamma$ in the accuracy of the Skellam mechanism and the factor $\sqrt{\log(2/\beta)\cdot\log(1/\delta)/\gamma}$ in the accuracy of the Geometric mechanism goes to $0$ when $\delta$ and $\beta$ go to $0$. For example, fix $S(f)=1, \delta=0.01, \alpha=50, \beta=0.1, \gamma=1$. Then the Geometric mechanism preserves $(\epsilon,\delta)$-\mbox{\upshape\sffamily CDP} with $\epsilon\approx 0.30$, while the Skellam mechanism preserves $(\epsilon,\delta)$-\mbox{\upshape\sffamily CDP} with $\epsilon\approx 0.15$. An empirical accuracy comparison between the mechanisms is shown in Figure \ref{accuracycomp}. We observe that the error of the Geometric and the Skellam mechanisms have a similar behaviour for both variables $\delta$ and $\gamma$, while the error of the Binomial mechanism is roughly three times larger. 
Finally, we are able to prove our main result, Theorem \ref{mainthm}, which follows from the preceding analyses.

\begin{proof}[Proof of Theorem $\ref{mainthm}$] The claim follows from Theorem \ref{cdptheorem} together with Theorem \ref{PSATHEOREM} (instantiated with the efficient construction in Example \ref{DDHEXM}) and from Theorem \ref{privthm} together with Theorem \ref{errorthm}.
\end{proof}

\section{Conclusions}

In this work we continued a line of research opened by the work of Shi et al. \cite{2}. By lowering the security definition of a PSA scheme, we were able to prove that a secure scheme (in this sense) can be built upon key-homomorphic weak PRFs. Based on the DDH assumption, we gave an instantiation of a secure PSA scheme. If the plaintext space is large enough, it has a substantially more efficient decryption algorithm than the scheme in \cite{2} at the cost of a slightly less efficient encryption algorithm, and achieves non-adaptive security in the standard model. Using the notion of computational differential privacy, we provided a connection between a secure PSA scheme and a mechanism preserving differential privacy by showing that a differentially private mechanism preserves computational differential privacy if it is executed through a secure PSA scheme. Moreover, we compared the accuracy of the Geometric, the Binomial and the Skellam mechanisms which preserve differential privacy and are suitable for an execution through a PSA scheme. While the practical performances of the Geometric and the Skellam mechanisms are equally better than the performance of the Binomial mechanism, we were able to provide a slightly better bound for the Skellam mechanism at high privacy levels.

\appendix




\section{Proof of Theorem \ref{PSATHEOREM}}\label{ptproof}

Let \textbf{game} $\boldsymbol 1$ be the security game from Definition \ref{securitygame} instantiated for the PSA scheme of Theorem \ref{PSATHEOREM}. We need to show that the advantage $\mu_{1}(\kappa)$ of a probabilistic polynomial-time adversary $\mathcal{T}_1$ in winning this game is negligible in the security parameter $\kappa$. We define the following intermediate \textbf{game} $\boldsymbol 2$ for a probabilistic polynomial-time adversary $\mathcal{T}_2$ and then show that winning \textbf{game} $\boldsymbol 1$ is at least as hard as winning \textbf{game} $\boldsymbol 2$.\\

\noindent\begin{description}
 \item\textbf{Setup.} The challenger runs the \text{\sffamily Setup} algorithm on input security parameter $\kappa$ and returns public parameters $\mbox{\upshape\sffamily pp}$, time-steps $T$ and secret keys $s,s_1,\ldots,s_n$ with $s=(\bigast_{i=1}^n s_i)^{-1}$. It sends $\kappa, \mbox{\upshape\sffamily pp},T,s$ to $\mathcal{T}_2$.
\item\textbf{Queries.} The challenger flips a random bit $b\in_R\{0,1\}$. $\mathcal{T}_2$ chooses $U=\{i_1,\ldots,i_u\}\subseteq[n]$ and sends it to the challenger which returns $(s_i)_{i\in[n]\setminus U}$. $\mathcal{T}_2$ is allowed to query $(i,t,x_i)$ with $i\in U, t\in T, x_i\in\widehat{\mathcal{D}}$ and the challenger returns the following: if $b=0$ it sends $\text{\upshape\sffamily F}_{s_i}(t)\cdot \varphi(x_i)$ to $\mathcal{T}_2$; if $b=1$ it chooses 
\begin{align*} & h_{1,t},\ldots,h_{u-1,t}\in_R G^\prime,\\
 & h_{u,t}:= \prod_{j=1}^u \text{\upshape\sffamily F}_{s_{i_j}}(t)\cdot\left(\prod_{j=1}^{u-1}h_{j,t}\right)^{-1}
\end{align*}
and sends $h_{i,t}\cdot \varphi(x_i)$ to $\mathcal{T}_2$.
\item\textbf{Challenge.} $\mathcal{T}_2$ chooses $t^*\in T$ such that no encryption query at $t^*$ was made and queries a tuple $(x_i)_{i\in U}$. If $b=0$ the challenger sends $(\text{\upshape\sffamily F}_{s_i}(t^*)\cdot \varphi(x_i))_{i\in U}$ to $\mathcal{T}_2$; if $b=1$ it chooses 
\begin{align*} & h_{1,t^*},\ldots,h_{u-1,t^*}\in_R G^\prime,\\
 & h_{u,t^*}:= \prod_{j=1}^u \text{\upshape\sffamily F}_{s_{i_j}}(t^*)\cdot\left(\prod_{j=1}^{u-1}h_{j,t^*}\right)^{-1}
\end{align*}
and sends $(h_{i,t^*}\cdot \varphi(x_i))_{i\in U}$ to $\mathcal{T}_2$.
\item\textbf{Queries.} $\mathcal{T}_2$ is allowed to make the same type of queries as before with the restriction that no encryption query at $t^*$ can be made.
\item\textbf{Guess.} $\mathcal{T}_2$ outputs a guess about $b$.
\end{description}
The adversary wins the game if it correctly guesses $b$.\\

\begin{Lem}\label{gameonetwo} Let $\kappa$ be a security parameter. Let $\mathcal{T}_1$ be an adversary in \textbf{\upshape game $\boldsymbol 1$} with advantage $\mu_{1}(\kappa)>\text{\upshape\sffamily neg}(\kappa)$. Then there exists an adversary $\mathcal{T}_2$ in \textbf{\upshape game $\boldsymbol 2$} with advantage $\mu_{2}(\kappa)>\text{\upshape\sffamily neg}(\kappa)$.
\end{Lem}
\begin{proof} Given a successful adversary $\mathcal{T}_1$ in \textbf{game} $\boldsymbol 1$ we construct a successful adversary $\mathcal{T}_2$ in \textbf{game} $\boldsymbol 2$ as follows:
\noindent\begin{description}
 \item\textbf{Setup.} Receive $\kappa, \mbox{\upshape\sffamily pp},T,s$ from the \textbf{game} $\boldsymbol 2$-challenger and send it to $\mathcal{T}_1$.
\item\textbf{Queries.} Flip a random bit $b\in_R\{0,1\}$. Receive $U=\{i_1,\ldots,i_u\}\subseteq[n]$ from $\mathcal{T}_1$ and send it to the challenger. Forward the obtained response $(s_i)_{i\in[n]\setminus U}$ to $\mathcal{T}_1$. Forward $\mathcal{T}_1$'s queries $(i,t,x_i)$ with $i\in U, t\in T, x_i\in\widehat{\mathcal{D}}$ to the challenger and forward the obtained response $c_{i,t}$ to $\mathcal{T}_1$.
\item\textbf{Challenge.} $\mathcal{T}_1$ chooses $t^*\in T$ such that no encryption query at $t^*$ was made and queries two different tuples $(x_i^{[0]})_{i\in U},(x_i^{[1]})_{i\in U}$ with $\sum_{i\in U}x_i^{[0]}=\sum_{i\in U}x_i^{[1]}$. Query $(x_i^{[b]})_{i\in U}$ to the challenger. Obtain the response $(c_{i,t^*})_{i\in U}$ and forward it to $\mathcal{T}_1$. 
\item\textbf{Queries.} $\mathcal{T}_1$ can make the same type of queries as before with the restriction that no encryption query at $t^*$ can be made.
\item\textbf{Guess.} $\mathcal{T}_1$ gives a guess about $b$. If the guess is correct, then output $0$; if not, output $1$.
\end{description}
If $\mathcal{T}_1$ has output the correct guess about $b$ then $\mathcal{T}_2$ can say with high confidence that the challenge ciphertexts were generated using a weak PRF and therefore outputs $0$. On the other hand, if $\mathcal{T}_1$'s guess was not correct, then $\mathcal{T}_2$ can say with high confidence that the challenge ciphertexts were generated using random values and it outputs $1$. Formally:\par\medskip

\noindent\textbf{Case $\boldsymbol 1$.} Let $(c_{i,t^*})_{i\in U}=(\text{\upshape\sffamily F}_{s_i}(t^*)\cdot \varphi(x_i^{[b]}))_{i\in U}$. Then $\mathcal{T}_2$ perfectly simulates \textbf{game} $\boldsymbol 1$ for $\mathcal{T}_1$ and the distribution of the ciphertexts is the same as in \textbf{game} $\boldsymbol 1$:
\begin{align*} \Pr[\mathcal{T}_2 \text{ outputs } 0] = & \frac{1}{2}(\Pr[\mathcal{T}_1 \text{ outputs } 0 \,|\, b=0]+\Pr[\mathcal{T}_1 \text{ outputs } 1 \,|\, b=1])\\
 = & \Pr[\mathcal{T}_1 \text{ wins \textbf{game} } \boldsymbol 1]\\
 = & \frac{1}{2} + \mu_{1}(\kappa).
\end{align*}\par\medskip

\noindent \textbf{Case $\boldsymbol 2$.} Let $(c_{i,t^*})_{i\in U}=(h_{i,t^*}\cdot \varphi(x_i^{[b]}))_{i\in U}$. Then the ciphertexts are random with the constraint 
\[\prod_{i\in U}c_{i,t^*} =\prod_{i\in U}\text{\upshape\sffamily F}_{s_i}(t^*)\cdot \varphi(x_i^{[b]})\]
such that decryption yields the same sum as in case $1$.
Because of the perfect security of the one-time pad the probability that $\mathcal{T}_1$ wins \textbf{game} $\boldsymbol 1$ is $1/2$ and
\begin{align*} \Pr[\mathcal{T}_2 \text{ outputs } 1]= & \frac{1}{2}(\Pr[\mathcal{T}_1 \text{ outputs } 1 \,|\, b=0]+\Pr[\mathcal{T}_1 \text{ outputs } 0 \,|\, b=1])\\
 = & \Pr[\mathcal{T}_1 \text{ loses \textbf{game} } \boldsymbol 1]\\
 = & \frac{1}{2}.
\end{align*}
Finally we obtain that the advantage of $\mathcal{T}_2$ in winning \textbf{game} $\boldsymbol 2$ is
\[\mu_{2}(\kappa)=\frac{1}{2}\mu_{1}(\kappa)>\text{\upshape\sffamily neg}(\kappa).\]
\end{proof}

\noindent For a probabilistic polynomial-time adversary $\mathcal{T}_3$, we define a new intermediate \textbf{game} $\boldsymbol 3$ out of \textbf{game} $\boldsymbol 2$ by just cancelling the plaintext dependence in each step of \textbf{game} $\boldsymbol 2$, i.e. in the encryption queries and in the challenge, instead of $(i,t,x_i)$ the adversary $\mathcal{T}_3$ now just queries $(i,t)$ and the challenger in \textbf{game} $\boldsymbol 3$ sends
\begin{align*} \text{\upshape\sffamily F}_{s_i}(t), & \mbox{ if } b=0,\\
h_{i,t}, & \mbox{ if } b=1
\end{align*}
to the adversary $\mathcal{T}_3$. The rest remains the same as in \textbf{game}~$\boldsymbol 2$.\\
It is easy to see that if there exists a successful adversary in \textbf{game}~$\boldsymbol 2$ then there is also a successful adversary in \textbf{game}~$\boldsymbol 3$.

\begin{Lem}\label{gametwothree} Let $\kappa$ be a security parameter. Let $\mathcal{T}_2$ be an adversary in \textbf{\upshape game $\boldsymbol 2$} with advantage $\mu_{2}(\kappa)>\text{\upshape\sffamily neg}(\kappa)$. Then there exists an adversary $\mathcal{T}_3$ in \textbf{\upshape game $\boldsymbol 3$} with advantage $\mu_{3}(\kappa)>\text{\upshape\sffamily neg}(\kappa)$.
\end{Lem}

\begin{Rem} For comparison to the proof of adaptive security by Shi et al. \text{\upshape \cite{2}} we emphasise that in the reduction from Aggregator Obliviousness to an intermediate problem (Proof of Theorem $1$ in \text{\upshape \cite{2}}) an adversary $\mathcal{B}$ has to compute the ciphertexts $c_i=g^{x_i}H(t)^{s_i}$ for all users $i\in[n]$ and for all (!) time-steps $t$, since $\mathcal{B}$ does not know in advance for which $i\in[n]$ it will have to use the PRF $H(t)^{s_i}$ and for which $i\in[n]$ it will have to use real random values. Thus, $\mathcal{B}$ has to program the random oracle $H$ in order to know for all $t$ the corresponding random number $z$ with $H(t)=g^z$ (where $g$ is a generator) for simulating the original Aggregator Obliviousness game. In contrast, in the reduction for our non-adaptive version of Aggregator Obliviousness, it is not necessary to program such an oracle, since the simulating adversary $\mathcal{T}_2$ knows in advance the set of non-compromised users and, for all (!) $t$, it can already decide for which $i\in[n]$ it will use the PRF (which in our case is $t^{s_i}$ instead of $H(t)^{s_i}$) and for which $i\in[n]$ it will use a real random value.
\end{Rem}

\noindent In the next step, the problem of distinguishing the weak PRF family 
\[\mathcal{F}=\{\text{\upshape\sffamily F}_s:M\to G^\prime\}_{s\in S}\] 
from a random function family has to be reduced to the problem of winning \textbf{\upshape game $\boldsymbol 3$}. We use a hybrid argument.

\begin{Lem}\label{gamethreeprf} Let $\kappa$ be a security parameter. Let $\mathcal{T}_3$ be an adversary in \textbf{\upshape game $\boldsymbol 3$} with advantage $\mu_{3}(\kappa)$. Then $\mu_{3}(\kappa)\leq\text{\upshape\sffamily neg}(\kappa)$ if 
\[\mathcal{F}=\{\text{\upshape\sffamily F}_s\,|\,\text{\upshape\sffamily F}_s:M\to G^\prime\}_{s\in S}\] 
is a weak PRF family.
\end{Lem}
\begin{proof} We define the following sequence of hybrid games, \textbf{game} $\boldsymbol 3_l$ with $l=1,\ldots,u-1$, for a probabilistic polynomial-time adversary $\mathcal{T}_{3_l}$.
\noindent\begin{description}
 \item\textbf{Setup.} As in \textbf{game} $\boldsymbol 3$.
\item\textbf{Queries.} The challenger flips a random bit $b\in_R\{0,1\}$. $\mathcal{T}_{3_l}$ chooses $U=\{i_1,\ldots,i_u\}\subseteq[n]$ and sends it to the challenger which returns $(s_i)_{i\in[n]\setminus U}$. $\mathcal{T}_{3_l}$ is allowed to query $(i,t)$ with $i\in U, t\in T$ and the challenger returns the following: if $i\notin \{i_1,\ldots,i_{l+b}\}$ it sends $\text{\upshape\sffamily F}_{s_i}(t)$ to $\mathcal{T}_{3_l}$; if $i\in \{i_1,\ldots,i_{l+b}\}$ it chooses 
\begin{align*} & h_{1,t},\ldots,h_{l-(1-b),t}\in_R G^\prime,\\
  & h_{l+b,t}:= \prod_{j=1}^{l+b} \text{\upshape\sffamily F}_{s_{i_j}}(t)\cdot\left(\prod_{j=1}^{l-(1-b)}h_{j,t}\right)^{-1}
\end{align*}
and sends $h_{i,t}$ to $\mathcal{T}_{3_l}$.
\item\textbf{Challenge.} $\mathcal{T}_{3_l}$ chooses $t^*\in T$ such that no encryption query at $t^*$ was made. The challenger chooses 
\begin{align*} & h_{1,t^*},\ldots,h_{l-(1-b),t^*}\in_R G^\prime,\\
 & h_{l+b,t^*}:= \prod_{j=1}^{l+b} \text{\upshape\sffamily F}_{s_{i_j}}(t^*)\cdot\left(\prod_{j=1}^{l-(1-b)}h_{j,t^*}\right)^{-1}
\end{align*}
and sends the following sequence to $\mathcal{T}_{3_l}$: 
\[(h_{1,t^*},\ldots,h_{l+b,t^*},\text{\upshape\sffamily F}_{s_{i_{l+b+1}}}(t^*),\ldots,\text{\upshape\sffamily F}_{s_{i_u}}(t^*)).\]
\item\textbf{Queries.} $\mathcal{T}_{3_l}$ can make the same type of queries as before with the restriction that no encryption query at $t^*$ can be made.
\item\textbf{Guess.} $\mathcal{T}_{3_l}$ outputs a guess about $b$.
\end{description}
The adversary wins the game if it correctly guesses $b$.\\

\noindent It is easy to see that \textbf{game} $\boldsymbol 3_1$ with $b=0$ corresponds to the case $b=0$ in \textbf{game} $\boldsymbol 3$ and \textbf{game} $\boldsymbol 3_{u-1}$ with $b=1$ corresponds to the case $b=1$ in \textbf{game} $\boldsymbol 3$. Moreover the ciphertexts in \textbf{game} $\boldsymbol 3_l$ with $b=1$ have the same distribution as the ciphertexts in \textbf{game} $\boldsymbol 3_{l+1}$ with $b=0$. Therefore
\[\Pr[\mathcal{T}_{3_{l+1}}\mbox{ wins \textbf{game }} \boldsymbol 3_{l+1}\,|\,b=0]=\Pr[\mathcal{T}_{3_l}\mbox{ loses \textbf{game }} \boldsymbol 3_l\,|\,b=1].\]
Using a successful adversary $\mathcal{T}_{3_l}$ in \textbf{\upshape game $\boldsymbol 3_l$} we construct a successful probabilistic polynomial-time distinguisher $\mathcal{D}_{\mbox{\scriptsize PRF}}$ which has access to an oracle 
\[\mathcal{O}(\cdot)\in_R\{\text{\upshape\sffamily F}_{s^\prime}(\cdot),\mbox{\upshape\sffamily rand}(\cdot)\},\mbox{ where }\] 
\[\text{\upshape\sffamily F}_{s^\prime}:M\to G^\prime\] 
is a weak PRF and 
\[\mbox{\upshape\sffamily rand}:M\to G^\prime\] 
is a real random function. $\mathcal{D}_{\mbox{\scriptsize PRF}}$ gets $\kappa$ as input and proceeds as follows.
\noindent\begin{enumerate}
\item Choose two indices $k_1,k_2\in[n]$ and guess that $k_1, k_2$ will be the $i_l^{\text{th}},i_{l+1}^{\text{th}}$ indices in $U$ specified by the adversary $\mathcal{T}_{3_l}$. This guess will be correct with probability $1/n^2$.
\item Choose $s\in_R S$,$s_i\in_R S$ for all $i\in[n]\setminus\{k_1,k_2\}$, generate $\mbox{\upshape\sffamily pp}$ and $T$ with $t\in_R M$ for all $t\in T$. Compute $\text{\upshape\sffamily F}_s(t)$ for all $t\in T$.
\item Make oracle queries for $t$ and receive $\mathcal{O}(t)$ for all $t\in T$.
\item Send $\kappa, \mbox{\upshape\sffamily pp},T, s$ to $\mathcal{T}_{3_l}$.
\item\textbf{Queries.} Receive $U=\{i_1,\ldots,i_u\}\subseteq[n]$ from $\mathcal{T}_{3_l}$. If $i_l\neq k_1$ or $i_{l+1}\neq k_2$ then abort. Else send $(s_i)_{i\in[n]\setminus U}$ to $\mathcal{T}_{3_l}$.
 If $\mathcal{T}_{3_l}$ queries $(i,t)$ with $i\in U, t\in T$ then return the following: if $i\notin \{i_1,\ldots,i_{l+1}\}$ send $\text{\upshape\sffamily F}_{s_i}(t)$ to $\mathcal{T}_{3_l}$; if $i=i_{l+1}=i_{k_2}$ send $\mathcal{O}(t)$ to $\mathcal{T}_{3_l}$; if $i\in \{i_1,\ldots,i_l\}$ choose 
\begin{align*} & h_{1,t},\ldots,h_{l-1,t}\in_R G^\prime,\\ 
 & h_{l,t}:=\left(\text{\upshape\sffamily F}_s(t)\cdot\mathcal{O}(t)\cdot\prod_{j=1}^{l-1}h_{j,t}\cdot\prod_{i\in[n]\setminus\{i_1,\ldots,i_{l+1}\}} \text{\upshape\sffamily F}_{s_i}(t)\right)^{-1}
\end{align*}
and send $h_{i,t}$ to $\mathcal{T}_{3_l}$.
\item\textbf{Challenge.} $\mathcal{T}_{3_l}$ chooses $t^*\in T$ such that no encryption query at $t^*$ was made. Choose 
\begin{align*} & h_{1,t^*},\ldots,h_{l-1,t^*}\in_R G^\prime,\\ 
 & h_{l,t^*}:=\left(\text{\upshape\sffamily F}_s(t^*)\cdot\mathcal{O}(t^*)\cdot\prod_{j=1}^{l-1}h_{j,t^*}
\cdot\prod_{i\in[n]\setminus\{i_1,\ldots,i_{l+1}\}} \text{\upshape\sffamily F}_{s_i}(t^*)\right)^{-1}
\end{align*}
and send the following sequence to $\mathcal{T}_{3_l}$: 
\[(h_{1,t^*},\ldots,h_{l,t^*},\mathcal{O}(t^*),\text{\upshape\sffamily F}_{s_{i_{l+2}}}(t^*),\ldots,\text{\upshape\sffamily F}_{s_{i_u}}(t^*)).\]
\item\textbf{Queries.} $\mathcal{T}_{3_l}$ can make the same type of queries as before with the restriction that no encryption query at $t^*$ can be made.
\item\textbf{Guess.} $\mathcal{T}_{3_l}$ outputs a guess about whether the $i_{l+1}^{\text{th}}$ element is random or pseudo-random. Output the same guess.\footnote{Essentially, here the specification of the set of non-compromised users before making any query allows $\mathcal{D}_{\mbox{\scriptsize PRF}}$ to be consistent with pseudo-random values or real random values in its replies to the queries.} 
\end{enumerate}
If $\mathcal{T}_{3_l}$ has output the correct guess about whether the $i_{l+1}^{\text{th}}$ element is random or pseudo-random then $\mathcal{D}_{\mbox{\scriptsize PRF}}$ can distinguish between $\text{\upshape\sffamily F}_{s^\prime}(\cdot)$ and $\mbox{\upshape\sffamily rand}(\cdot)$. Now we prove this result formally and show that, in this way, \textbf{game} $\boldsymbol 3_l$ is perfectly simulated by $\mathcal{T}_{3_l}$.\par\medskip

\noindent\textbf{Case $\boldsymbol 1$.} Let $\mathcal{O}(\cdot)=\text{\upshape\sffamily F}_{s^\prime}(\cdot)$. Define $s_{i_{l+1}}:=s^\prime$. Since $S,M$ are groups, there exists an element $s_{i_l}$ with 
\[s_{i_l}=(s\ast\bigast_{i\in[n]\setminus\{i_l\}} s_i)^{-1}\]
and for all $t\in T$:
\[\left(\text{\upshape\sffamily F}_s(t)\cdot \text{\upshape\sffamily F}_{s^\prime}(t)\cdot\prod_{i\in[n]\setminus\{i_1,\ldots,i_{l+1}\}} \text{\upshape\sffamily F}_{s_i}(t)\right)^{-1} = \prod_{j=1}^l \text{\upshape\sffamily F}_{s_{i_j}}(t).\]
Then for all $t\in T$ the value $h_{l,t}$ is equal to
\[\left(\text{\upshape\sffamily F}_s(t)\cdot\prod_{j=1}^{l-1}h_{j,t}\cdot \text{\upshape\sffamily F}_{s^\prime}(t)\cdot\prod_{i\in[n]\setminus\{i_1,\ldots,i_{l+1}\}} \text{\upshape\sffamily F}_{s_i}(t)\right)^{-1}=\prod_{j=1}^l \text{\upshape\sffamily F}_{s_{i_j}}(t)\cdot\left(\prod_{j=1}^{l-1}h_{j,t}\right)^{-1}.\]
Therefore the distribution of the ciphertexts corresponds exactly to the case in \textbf{game} $\boldsymbol 3_l$ with $b=0$.\\

\noindent\textbf{Case $\boldsymbol 2$.} Let $\mathcal{O}(\cdot)=\mbox{\upshape\sffamily rand}(\cdot)$. Define the random elements 
\[h_{l+1,t}:=\mbox{\upshape\sffamily rand}(t)\] 
for all $t\in T$. Since $S,M$ are groups, there exists an element $s^\prime\in S$ with 
\[s^\prime=(s\ast\bigast_{i\in[n]\setminus\{i_l,i_{l+1}\}} s_i)^{-1}.\] 
Let $s_{i_l}\in_R S$ and $s_{i_{l+1}}:=s^\prime * s_{i_l}^{-1}$. Then for all $t\in T$:
\[\left(\text{\upshape\sffamily F}_s(t)\cdot\prod_{i\in[n]\setminus\{i_1,\ldots,i_{l+1}\}} \text{\upshape\sffamily F}_{s_i}(t)\right)^{-1} = \prod_{j=1}^{l+1} \text{\upshape\sffamily F}_{s_{i_j}}(t)\] 
and the value $h_{l,t}$ is equal to
\[\left(\text{\upshape\sffamily F}_s(t)\cdot h_{l+1,t}\cdot\prod_{j=1}^{l-1}h_{j,t}\cdot\prod_{i\in[n]\setminus\{i_1,\ldots,i_{l+1}\}} \text{\upshape\sffamily F}_{s_i}(t)\right)^{-1}\]
and equivalently 
\[h_{l+1,t} = \prod_{j=1}^{l+1} \text{\upshape\sffamily F}_{s_{i_j}}(t)\cdot\left(\prod_{j=1}^{l}h_{j,t}\right)^{-1}.\]
Therefore the distribution of the ciphertexts corresponds exactly to the case in \textbf{game} $\boldsymbol 3_l$ with $b=1$.\par\medskip

\noindent Without loss of generality, let 
\[\Pr[\mathcal{T}_{3_l}\mbox{ wins \textbf{game }} \boldsymbol 3_l\,|\,b=0]\geq \Pr[\mathcal{T}_{3_l}\mbox{ loses \textbf{game }} \boldsymbol 3_l\,|\,b=1].\]
All in all, we obtain
\begin{align*} & \Pr[\mathcal{T}_{3_l}\mbox{ wins \textbf{game }} \boldsymbol 3_l\,|\,b=0]-\Pr[\mathcal{T}_{3_l}\mbox{ loses \textbf{game }} \boldsymbol 3_l\,|\,b=1]\\
 = & \Pr[\mathcal{D}_{\mbox{\scriptsize PRF}}^{\text{\upshape\sffamily F}_{s^\prime}(\cdot)}(\kappa)=1\,|\,i_l=k_1,i_{l+1}=k_2]-\Pr[\mathcal{D}_{\mbox{\scriptsize PRF}}^{\mbox{\scriptsize\upshape\sffamily rand}(\cdot)}(\kappa)\,\,\,=1\,|\,i_l=k_1,i_{l+1}=k_2]\\
 \leq & \,n^2\cdot (\Pr[\mathcal{D}_{\mbox{\scriptsize PRF}}^{\text{\upshape\sffamily F}_{s^\prime}(\cdot)}(\kappa)=1]-\Pr[\mathcal{D}_{\mbox{\scriptsize PRF}}^{\mbox{\scriptsize\upshape\sffamily rand}(\cdot)}(\kappa)=1])
\end{align*}
and since $n$ is polynomial in $\kappa$, this expression is negligible by the pseudo-randomness of $\text{\upshape\sffamily F}_{s^\prime}(\cdot)$ on uniformly chosen input. Therefore, the advantage of $\mathcal{T}_{3_l}$ in winning \textbf{game} $\boldsymbol 3_l$ is negligible.\\ 
Finally, by a hybrid argument we have:
\begin{align*} & \Pr[\mathcal{T}_{3}\mbox{ wins \textbf{game }} \boldsymbol 3]\\
 = & \frac{1}{2}(\Pr[\mathcal{T}_{3}\mbox{ wins \textbf{game }} \boldsymbol 3\,|\,b=0]+\Pr[\mathcal{T}_{3}\mbox{ wins \textbf{game }} \boldsymbol 3\,|\,b=1])\\
 = & \frac{1}{2}(\Pr[\mathcal{T}_{3_1}\mbox{ wins \textbf{game }} \boldsymbol 3_1\,|\,b=0]+\Pr[\mathcal{T}_{3_{u-1}}\mbox{ wins \textbf{game }} \boldsymbol 3_{u-1}\,|\,b=1])\\
 = & \frac{1}{2}+\frac{1}{2}(\Pr[\mathcal{T}_{3_1}\mbox{ wins \textbf{game }} \boldsymbol 3_1\,|\,b=0]-\Pr[\mathcal{T}_{3_{u-1}}\mbox{ loses \textbf{game }} \boldsymbol 3_{u-1}\,|\,b=1])\\
 = & \frac{1}{2}+\frac{1}{2}\sum_{l=1}^{u-1}\Pr[\mathcal{T}_{3_l}\mbox{ wins \textbf{game }} \boldsymbol 3_l\,|\,b=0]-\Pr[\mathcal{T}_{3_l}\mbox{ loses \textbf{game }} \boldsymbol 3_l\,|\,b=1]\\
 = & \frac{1}{2}+\text{\upshape\sffamily neg}(\kappa).
\end{align*}
\end{proof}\par\medskip

\noindent We can now complete the proof of Theorem \ref{PSATHEOREM}.

\begin{proof}[Proof of Theorem $\ref{PSATHEOREM}$]  By Lemma \ref{gameonetwo} - \ref{gamethreeprf}:\\
$\mu_1(\kappa)=2\cdot\mu_2(\kappa)=2\cdot\mu_3(\kappa)=2\cdot(u-1)\cdot n^2\cdot\text{\upshape\sffamily neg}(\kappa)<2\cdot n^3\cdot\text{\upshape\sffamily neg}(\kappa)=\text{\upshape\sffamily neg}(\kappa)$.
\end{proof}



\section{The Skellam mechanism}\label{skellamsec}

\subsection{Preliminaries}

As observed before, the distributed noise generation is feasible with a probability distribution function closed under convolution. For this purpose, we recall the Skellam distribution.

\begin{Def}[Skellam Distribution \cite{29}]\label{skellam} Let $\mu_1$, $\mu_2> 0$. A discrete random variable $X$ is drawn according to the Skellam distribution with parameters $\mu_1,\mu_2$ (short: $X\sim \text{\upshape Sk}(\mu_1,\mu_2)$) if it has the following probability distribution function $\psi_{\mu_1,\mu_2}\colon\mathbb{Z}\mapsto\mathbb{R}$:
\[\psi_{\mu_1,\mu_2}(k)=e^{-(\mu_1+\mu_2)}\left(\frac{\mu_1}{\mu_2}\right)^{k/2}I_k(2\sqrt{\mu_1\mu_2}),\]
where $I_k$ is the modified Bessel function of the first kind (see pages $374$--$378$ in $\cite{28}$).
\end{Def}

A random variable $X\sim \text{\upshape Sk}(\mu_1,\mu_2)$ has variance $\mu_1+\mu_2$ and can be generated as the difference of two random variables drawn according to the Poisson distribution of mean $\mu_1$ and $\mu_2$, respectively \cite{29}. Note that the Skellam distribution is not generally symmetric. However, we mainly consider the particular case $\mu_1=\mu_2=\mu/2$ and refer to this symmetric distribution as $\text{Sk}(\mu) = \text{Sk}(\mu/2,\mu/2)$.

\begin{Lem}[\cite{29}]\label{sksum} Let $X\sim \text{\upshape Sk}(\mu_1,\mu_2)$ and $Y\sim \text{\upshape Sk}(\mu_3,\mu_4)$ be independent random variables. Then $Z:=X+Y$ is distributed according to $\text{\upshape Sk}(\mu_1+\mu_3,\mu_2+\mu_4)$.
\end{Lem}

An induction step shows that the sum of $n$ i.i.d. symmetric Skellam random variables with variance $\mu$ is a symmetric Skellam random variable with variance $n\mu$. Suppose that adding symmetric Skellam noise with variance $\mu$ preserves $(\epsilon,\delta)$-\mbox{\upshape\sffamily DP}. Recall that the network is given an a priori known estimate $\gamma$ of the lower bound on the fraction of non-compromised users. We define $\mu_{user}=\mu/(\gamma n)$ and instruct the users to add symmetric Skellam noise with variance $\mu_{user}$ to their own data. If compromised users will not add noise, the total noise will be still sufficient to preserve $(\epsilon,\delta)$-\mbox{\upshape\sffamily DP}.\\ 
For our analysis, we will use the following bound on the ratio of modified Bessel functions of the first kind.

\begin{Lem}[\cite{27}]\label{modbesrat} For real $k>0$ let $I_k(\mu)$ be the modified Bessel function of the first kind and order $k$. Then
 \[\frac{I_k(\mu)}{I_{k+1}(\mu)}<\frac{\mu}{-(k+1)+\sqrt{(k+1)^2+\mu^2}}.\]
\end{Lem}

For the privacy analysis of the Skellam mechanism, we need a tail bound on the symmetric Skellam distribution.

\begin{Lem}\label{SKELLAMBOUND} Let $X\sim \text{\upshape Sk}(\mu)$ and let $\sigma>0$. Then, for all $\tau\geq -\sigma\mu$, 
\begin{align*}
  & \Pr[X>\sigma\mu + \tau] \\
  \leq & \ e^{-\mu\left(1-\sqrt{1+\sigma^2}+\sigma\ln(\sigma + \sqrt{1+\sigma^2})\right)-\tau\ln(\sigma + \sqrt{1+\sigma^2})}.   
\end{align*}
\end{Lem}
\begin{proof}
We use standard techniques from probability theory. Applying Markov's inequality, for any $t>0$,
\begin{align*}\Pr[X>\sigma\mu + \tau]= & \Pr[e^{tX}>e^{t(\sigma\mu + \tau)}]\\
\leq & \frac{\E[e^{tX}]}{e^{t(\sigma\mu + \tau)}}.
\end{align*}
As shown in \cite{30}, for $X\sim \text{Sk}(\mu)$, the moment generating function of $X$ is 
\[\E[e^{tX}]=e^{-\mu(1-\cosh(t))},\]
where $\cosh(t)= (e^t + e^{-t})/2$. Hence, we have
\[\Pr[X>\sigma\mu + \tau]\leq e^{-\mu(1-\cosh(t)+t\sigma)-t\tau}.\]
Fix $t=\ln(\sigma + \sqrt{1+\sigma^2})$. In order to conclude the proof, we observe that $\cosh(\ln(\sigma + \sqrt{1+\sigma^2}))=\sqrt{1+\sigma^2}$.
\end{proof}

\noindent One can easily verify that, for $\sigma>0$,
\[1-\sqrt{1+\sigma^2}+\sigma\ln(\sigma + \sqrt{1+\sigma^2})> 0.\]

\subsection{Analysis of the Skellam mechanism}

In this section, we provide a bound on the variance $\mu$ of the symmetric Skellam distribution (as stated in Theorem \ref{privthm}) that is needed in order to preserve $(\epsilon,\delta)$-differential privacy and we compute the error that is thus introduced.\par\bigskip\bigskip

\noindent\textbf{Privacy analysis}


\begin{Thm}\label{skmech} Let $\epsilon>0$ and let $0<\delta<1$. For all databases $D\in\mathcal{D}^n$ the randomised mechanism \[\mathcal{A}_{Sk}(D):=f(D)+Y\]
preserves $(\epsilon,\delta)$-\mbox{\upshape\sffamily DP} with respect to any query $f$ of sensitivity $S(f)$, where $Y\sim\text{\upshape Sk}(\mu)$ with 
\[\mu\geq\frac{\log(1/\delta)}{1-\cosh(\epsilon/S(f))+(\epsilon/S(f))\cdot\sinh(\epsilon/S(f))}.\]
\end{Thm}
\begin{proof} Let $D_0, D_1\in\mathcal{D}^n$ be adjacent databases with $|f(D_0)-f(D_1)|\leq S(f)$. The largest ratio between $\Pr[\mathcal{A}_{Sk}(D_0)=R]$ and $\Pr[\mathcal{A}_{Sk}(D_1)=R]$ is reached when $k:=R-f(D_0)=R-f(D_1)-S(f)\geq 0$, where $R$ is any possible output of $\mathcal{A}_{Sk}$.
Then, by Lemma \ref{modbesrat}, for all possible outputs $R$ of $\mathcal{A}_{Sk}$:
\begin{align*} \frac{\Pr[\mathcal{A}_{Sk}(D_0)=R]}{\Pr[\mathcal{A}_{Sk}(D_1)=R]}= & \frac{\Pr[Y=k]}{\Pr[Y=k+S(f)]}\\
 = & \prod_{j=1}^{S(f)}\frac{\Pr[Y=k+j-1]}{\Pr[Y=k+j]}\\
 < & \prod_{j=1}^{S(f)}\frac{\mu}{-(k+j)+\sqrt{(k+j)^2+\mu^2}}\\
 \leq & e^\epsilon.\numberthis\label{skmechdp}
\end{align*}
Inequality \eqref{skmechdp} holds if $k\leq\sinh(\epsilon/S(f))\cdot\mu-S(f)$, since it implies $k\leq\sinh(\epsilon/S(f))\cdot\mu-j$ for all $j=1,\ldots,S(f)$ and
\[\frac{\mu}{-(k+j)+\sqrt{(k+j)^2+\mu^2}}\leq e^{\epsilon/S(f)}.\] 
Applying Lemma \ref{SKELLAMBOUND} with $\sigma=\sinh(\epsilon/S(f))$ and $\tau=-S(f)$, we get
\begin{align*} & \Pr[k>\sinh(\epsilon/S(f))\cdot\mu-S(f)]\\
\leq & e^{-\mu\cdot(1-\cosh(\epsilon/S(f))+(\epsilon/S(f))\cdot\sinh(\epsilon/S(f)))+\epsilon}
\end{align*}
and this expression is set to be smaller or equal than $\delta$. This inequality is satisfied if
\[\mu\geq\frac{\log(1/\delta)}{1-\cosh(\epsilon/S(f))+(\epsilon/S(f))\cdot\sinh(\epsilon/S(f))}.\]
\end{proof}

\begin{Rem} The bound on $\mu$ from Theorem \text{\upshape \ref{skmech}} is smaller than $2\cdot (S(f)/\epsilon)^2\cdot\log(1/\delta)$, thus the standard deviation of $Y\sim\text{\upshape Sk}(\mu)$ is linear in $S(f)/\epsilon$ (for constant $\delta$).
\end{Rem}\par\bigskip\bigskip

\noindent\textbf{Accuracy analysis}

\begin{Thm}\label{erroranalysis} Let $\epsilon>0$ and $0<\delta<1$. Then for all $0<\beta<1$ the mechanism specified in Theorem $\ref{skmech}$ has $(\alpha,\beta)$-accuracy, where
\[\alpha=\frac{S(f)}{\epsilon}\cdot\left(\log\left(\frac{2}{\beta}\right)+\log\left(\frac{1}{\delta}\right)\right).\]
\end{Thm}
\begin{proof} Let $\mu=\frac{\log(1/\delta)}{1-\cosh(\epsilon/S(f))+(\epsilon/S(f))\cdot\sinh(\epsilon/S(f))}$ be the bound on the variance for the Skellam mechanism provided in Theorem \ref{skmech}. Now, as in the proof of Lemma \ref{SKELLAMBOUND}, for $\alpha^\prime>0$,
\begin{align*} \Pr[|N|>\alpha^\prime]= & 2\cdot\Pr[N>\alpha^\prime]\\
\leq & 2\cdot e^{-\mu\cdot(1-\cosh(\epsilon/S(f)))-(\epsilon/S(f))\cdot\alpha^\prime}
\end{align*}
and this expression is set to be equal to $\beta$. Solving this equality for $\alpha^\prime$ yields
\begin{align*} \alpha^\prime= & \frac{S(f)}{\epsilon}\cdot\left(\log\left(\frac{2}{\beta}\right)+\left(\cosh\left(\frac{\epsilon}{S(f)}\right)-1\right)\cdot\mu\right)\\
 \leq & \frac{S(f)}{\epsilon}\cdot\left(\log\left(\frac{2}{\beta}\right)+\log\left(\frac{1}{\delta}\right)\right)\\
= & \alpha.
\end{align*}
\end{proof}\par\medskip

\noindent For the distributed noise generation, each single user adds symmetric Skellam noise with variance $\mu_{user}=\mu/(\gamma n)$ to her data. The worst case for accuracy is when all $n$ users add noise, thus the total noise $N$ is a symmetric Skellam variable with variance $\mu/\gamma$ and the accuracy becomes
\[\alpha=\frac{S(f)}{\epsilon}\cdot\left(\log\left(\frac{2}{\beta}\right)+\frac{1}{\gamma}\cdot\log\left(\frac{1}{\delta}\right)\right),\]
proving the third claim of Theorem \ref{errorthm}.

\addcontentsline{toc}{chapter}{Literatur}
\bibliography{Literatur}

\end{document}